\documentclass{LMCS}

\def\eqalign#1{\null\,\vcenter{\openup\jot\mathsurround=0 pt
  \ialign{\strut\hfil$\displaystyle{##}$&$\displaystyle{{}##}$\hfil
      \crcr#1\crcr}}\,}

\usepackage{logiclmcs}
\usepackage{amssymb}
\usepackage{amsmath}
\usepackage{graphicx}
\usepackage{enumerate,hyperref}
\newcommand{\ign}[1]{}

\newcommand{\miff}{\text{ iff }}

\newcommand{\tlef}{\mathsf{EF}}

\newcommand{\tlf}{\mathsf{F}}
\newcommand{\tlx}{\mathsf{X}}
\newcommand{\ctlex}{\mathsf{EX}}

\newcommand{\freef}[2]{(#1,#2)^\Delta}
\newcommand{\onefreef}[1]{#1^\Delta}

\newcommand{\hole}{\Box}

\newcommand{\mainle}{\dashv}
\newcommand{\fotwo}{\tlef + \tlf^{-1}}
\newcommand{\wordfotwo}{\tlf + \tlf^{-1}}

\def\doi{5 (3:5) 2009}
\lmcsheading%
{\doi}
{1--29}
{}
{}
{Jul.~16, 2008}
{Aug.~\phantom{0}5, 2009}
{} 

\begin{document}

\title[Two-way unary temporal logic over trees]%
  {Two-way unary temporal logic over trees\rsuper*}
\author[M.~Boja{\'n}czyk]{Miko{\l}aj Boja{\'n}czyk}
\address{Warsaw University}
\email{bojan@mimuw.edu.pl}

\titlecomment{{\lsuper*}Supported by Polish government grant no.\ 
  N206 008 32/0810.}

\begin{abstract}
  We consider a temporal logic $\fotwo$ for unranked, unordered finite
  trees.  The logic has two operators: $\tlef \varphi$, which says
  ``in some proper descendant $\varphi$ holds'', and $\tlf^{-1}
  \varphi$, which says ``in some proper ancestor $\varphi$ holds''. We
  present an algorithm for deciding if a regular language of unranked
  finite trees can be expressed in $\fotwo$. The algorithm uses a
  characterization expressed in terms of forest algebras.
\end{abstract}
\maketitle

\section{Introduction}
We say a logic has a decidable characterization if the following
decision problem is decidable: ``given as input a finite automaton,
decide if the recognized language can be defined using a formula of
the logic''. Representing the input language by a finite automaton is
a reasonable choice, since many known logics (over words or trees) are
captured by finite automata. 

This type of problem has been successfully studied for word languages.
Arguably best known is the result of McNaughton, Papert and
Sch\"utzenberger~\cite{schutzenberger, mcnaughton}, which says that
the following three conditions on a regular word language $L$ are
equivalent: a) $L$ can be defined in first-order logic; b) $L$ can be
defined using a star-free expression; and c) the syntactic semigroup
of $L$ does not contain a non-trivial group. Since condition c) can be
effectively tested, the above theorem gives a decidable
characterization of first-order logic. This result demonstrates two
important features of work in this field: a decidable characterization
not only gives a better understanding of the logic in question, but it
often reveals unexpected connections with algebraic concepts.  During
several decades of research, decidable characterizations have been
found for fragments of first-order logic with restricted
quantification and a large group of temporal logics, see
\cite{pin-survey} and \cite{wilke} for references.

For trees, however, much less is known. No decidable characterization
has been found for what is possibly the most important subclass of
regular tree languages, first-order logic with the descendant
relation, despite several attempts~\cite{potthoff,heuter,doktorat}.
Similarly open are chain logic~\cite{chainlogic} and the temporal
logics CTL, CTL* and PDL.  However, there has been some recent
progress.  In~\cite{efextcs}, decidable characterizations were
presented for the temporal logics $\tlef$ and $\ctlex + \tlef$; while
Benedikt and Segoufin~\cite{segoufinfo} characterized tree languages
definable in first-order logic with the successor relation (but
without the descendant relation). Two new results give effective
characterizations for some fragments of first-order logic with limited
quantifier alternation. The expressive power of alternation-free
formulas (i.e.~boolean combinations of formulas with quantifier prefix
$\exists^*$) is characterized in~\cite{treepiecewise}. Properties
that can be defined both with quantifier prefix $\exists^*\forall^*$
and also with quantifier prefix $\forall^*\exists^*$ are characterized
in~\cite{deltatwo}.  We will come back to the latter class later on in
this introduction.

In this paper, we continue the line of research started
in~\cite{efextcs}, by focusing on a temporal logic for trees. We consider
a logic called $\fotwo$. This logic has two operators: $\tlef
\varphi$, which says ``in some proper descendant $\varphi$ holds'',
and $\tlf^{-1} \varphi$, which says ``in some proper ancestor
$\varphi$ holds''. Thanks to the backward modality, $\fotwo$ is more
expressive than $\tlef$ alone. For instance, the formula
\[
  \tlef (a \land \neg
  \tlf^{-1} \neg b)
\]
defines the class of trees where  some node has label $a$, but all of its
ancestors have label $b$. This is a property reminiscent of CTL, and
cannot be expressed by only using $\tlef$, since it fails the
identities  that must be satisfied by $\tlef$-definable languages~\cite{forestalgebra}.

The main result in this paper is Theorem~\ref{thm:main}, which gives a
decidable characterization of languages definable in
$\fotwo$. Before we present this result, in
Section~\ref{sec:why-two-way} we try to justify the choice of the
logic $\fotwo$. In Section~\ref{sec:basic-definitions}
we present the algebraic formalism that will be used in the proofs.
The rest of the paper is devoted to proving the main result.

I would like to thank Luc Segoufin. We spent a lot of time together
trying to understand the expressive power of $\fotwo$. Without his
input this paper would not have been possible. I would also like to
thank the anonymous referees for their helpful comments.

\section{Why two-way  unary  temporal logic}
\label{sec:why-two-way}
There are two reasons to consider $\fotwo$. The first reason is that,
over words, this logic corresponds to an important and well-studied
class of regular languages. The second reason is that, over trees, the
logic is related to XML. We go over these reasons in
Sections~\ref{sec:word-analogy} and~\ref{sec:xpath} respectively.

\subsection{The word analogy}
\label{sec:word-analogy}

There is a very robust class of regular word languages that has
several equivalent descriptions (a survey of this class can be found
in \cite{dadiamond}):
\begin{enumerate}[(1)]
\item \label{item:wordtl} Word languages that can be defined in the temporal logic
  $\wordfotwo$. Here $\tlf \varphi$ means ``in some future
  position $\varphi$'' and $\tlf^{-1} \varphi$ means ``in some past
  position $\varphi$''.
\item \label{item:wordfo2} Word languages that can be defined by a
  first-order formula with two variables and the left-to-right
  ordering of positions (but without the successor relation).
\item \label{item:wordsigma} Word languages that can be defined by a
  first-order formula (with many variables, the left-to-right
  ordering, but without the successor relation) with a $\forall^* \exists^*$
  quantifier prefix, and also by one with an $\exists^* \forall^*$
  quantifier prefix.
\item \label{item:da} Word languages whose syntactic semigroup belongs to the
  semigroup variety DA.  One way of defining this variety is in terms
  of an identity: DA is the class of semigroups that satisfy the
  identity     $(st)^\omega = (st)^\omega s (st)^\omega$.
\item \label{item:unambig} Word languages described by finite disjoint
      unions of unambiguous products (a form of regular expression).
\item \label{item:turtleautomata} Word languages that can be
  recognized by ``turtle automata'', a type of deterministic two-way
  word automaton.
\item \label{item:twoway} Word languages that can be recognized by
  two-way deterministic automata where the states in a run are
  non-decreasing with respect to a given order.
\end{enumerate}
An important corollary of property \ref{item:da} is that membership of
a regular language in the above class is decidable: it suffices to
check if the syntactic semigroup of the language satisfies the DA
identity.

Some of the above classes generalize easily to trees, some don't.

We will not talk about classes~\ref{item:unambig},
\ref{item:turtleautomata} and \ref{item:twoway}. It is not clear what
unambiguous expressions are for trees, likewise for the automata.  

We will come back to the algebraic description in item~\ref{item:da}
later on in the paper.

The three logically defined classes~\ref{item:wordtl},
\ref{item:wordfo2} and \ref{item:wordsigma} can be easily extended to
trees.  A natural counterpart of class~\ref{item:wordtl} is the logic
$\fotwo$ considered in this paper. The classes~\ref{item:wordfo2} and
\ref{item:wordsigma} can define tree languages if the order is
interpreted as the ancestor/descendant ordering of tree nodes. (One
could also consider variants where two partial orders of nodes are
available instead of one: the ancestor/descendant order and also the
left-to-right ordering of siblings.  We keep to the simpler case,
where siblings are unordered.)  The logically defined classes diverge
for trees:
\begin{enumerate}[$\bullet$]
\item Two-variable logic is strictly more expressive than the temporal
      logic.  The translation from temporal to two-variable logic is
      fairly obvious.  For the converse, the problem is that $x \not
      \le y \land y \not \le x$ cannot be expressed in the temporal
      logic.  For instance, the language: ``there are two $a$'s'' can
      be defined by a two-variable formula, but cannot be defined in
      the temporal logic.  This is because the temporal logic is
      bisimulation invariant, and cannot see the difference between
      one child with $a$ and two children with $a$. (Note however,
      that the languages ``two $a$'s below some $b$'', or ``three
      $a$'s'' cannot be defined in two-variable logic.)
\item As we will show at the end of this paper, the intersection of
  $\forall^* \exists^*$ and $\exists^* \forall^*$ is incomparable with
  both the two-variable and the temporal logic. 
\end{enumerate}

The second fragment has been considered in~\cite{deltatwo}, the
investigation therein shows that it is a well-behaved class of tree
languages.  We are left with the temporal logic and two-variable
logic. Why do we choose temporal logic and not two-variable logic?
The reason is that two-variable logic seems to be less robust for
trees: why can ``two $a$'s'' be defined, but not ``three $a$'s''? Of
course it is nonetheless important to understand two-variable logic,
and we leave this task as future work.

\subsection{XPath}
\label{sec:xpath}
XPath is a formalism used to describe paths and nodes in unranked
trees.  There is a strong connection between XPath and two-variable
logics

A set of paths is seen as a binary relation $P(x,y)$, which says when
a source $x$ can be connected with a target $y$.  The basic idea in
XPath is that one starts with atomic paths, called axes, such as ``$x$
is a descendant of $y$'', or ``$x$ is a child of $y$'', and then
constructs longer paths using mechanisms such as concatenation.
Marx and de Rijke~\cite{marxderijke} show that a fragment of XPath called Core XPath
has exactly the same expressive power as two-variable first-order
logic. (The equivalence in expressive power is for Boolean queries  in XPath and
sentences of two-variable  logic. The equivalence also
holds for unary queries in XPath and formulas of two-variable logic
with one free variable; but  it fails for binary queries.)
 Note however, that the axes considered by Marx include child
and next-child, which go beyond the fragments considered in this
paper.  When the only axes allowed are ``descendant'' and
``ancestor'', Core XPath has exactly the same power as ``our'' logic
$\fotwo$.  A decidable characterization for fragments of XPath with
the other axes, including the one considered by Marx, is left as
future work.

\section{Basic definitions}
\label{sec:basic-definitions}
\subsection{Trees and forests}
\label{sec:trees-forests}
We work with unranked finite labeled trees. We assume that an alphabet
$(A,B)$ contains two types of labels: one set of labels $A$ that can
be used in the leaves, and another set of labels $B$ that can be used
in inner nodes (i.e.~not leaves). This division is convenient for the
algebraic framework we use in general, and for the induction proof in
this paper in particular.  \emph{Trees} are defined as follows: every
leaf label $a \in A$ is a tree; if $t_1,\ldots,t_n$ are trees and $b
\in B$ is an inner node label then $b(t_1 + \cdots + t_n)$ is a tree.
A \emph{forest} is a sequence of trees.  As above, we concatenate
forests using~$+$.  In particular every forest is of the form $ t =
t_1 + \cdots + t_n$, for some trees $t_1,\ldots,t_n$.  We do not allow
empty forests, so $n \ge 1$.  We denote both trees and forests using
letters $s,t$. When $b$ is a label and $t$ is a forest, we write $bt$
for the tree that has label $b$ in the root, and where the children
form the forest $t$. In other words, we omit the parentheses and write
$bt$ instead of $b(t)$.

A \emph{context} is a forest where exactly one leaf is labeled by
a special label $\hole$; this leaf is interpreted as a hole. We denote
contexts by $p, q$.  The \emph{main path} in a context
consists of the ancestors of the hole. A forest $ t$ can be
substituted in place of the hole of a context $p$, the resulting
forest is denoted $p( t)$, or sometimes $p t$.
\medskip

  \begin{center}
    \includegraphics[scale=0.65]{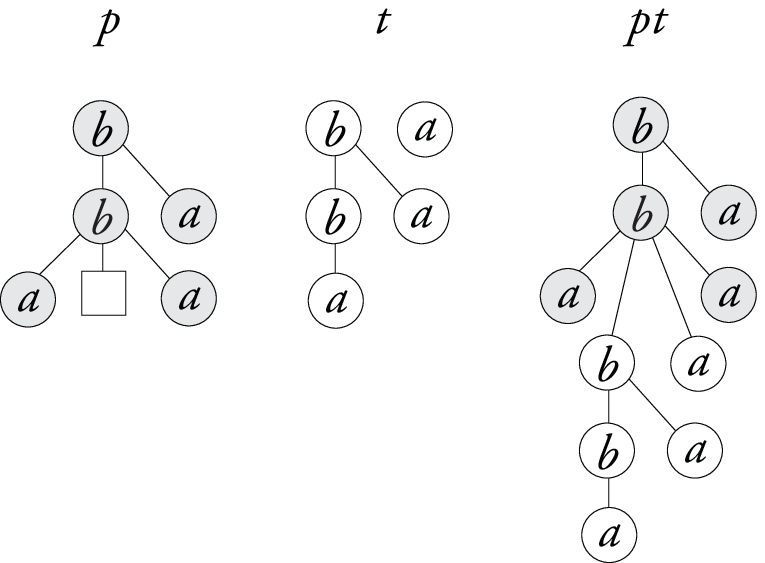}
  \end{center}

There is a natural composition operation on contexts: the context $p
q$ is the unique context such that $ (p q) t = p (q t) $ holds for all
forests~$ t$.  We allow the empty context, denoted by $\hole$; this is
the context where the only node in the context is the hole~$\hole$.
The empty context satisfies $\hole t = t$.  Nodes of trees, forests
and contexts are defined the usual way. We write $x,y,z$ for nodes,
and $x \le y$ when $x$ is an ancestor of $y$.

The reader will notice that the trees and forests we defined are
sibling-ordered (i.e.~$s+t$ is not the same as $t+s$).  However,
properties definable in our logic $\fotwo$ are going to
be invariant under this order.

\subsection{The logic}
\label{sec:logic}
The logic $\fotwo$ is defined as follows:
\begin{enumerate}[$\bullet$]
\item Every label -- both inner node label and leaf label -- is a
  formula; this formula holds in nodes with that label.
\item Formulas are closed under boolean combinations,
  including negation.
\item If $\varphi$ is a formula, then $\tlef \varphi$ is also a
  formula; it is true in a node $x$ if there is some proper descendant
  $y > x$ where $\varphi$ is true. Likewise for $\tlf^{-1} \varphi$,
  but this time $y$ must be a proper ancestor $y < x$.
\end{enumerate}

A formula $\varphi$ of $\fotwo$ is most naturally interpreted as a
unary query, i.e. in a given tree it selects a set of nodes. For
instance, the formula $\tlef \true$ selects all inner nodes. In this
paper, we are interested in tree languages, i.e.~boolean queries,
where a formula is either true or false in a given tree. To get a
boolean query, we say a formula of $\fotwo$ is true in a tree if it is
true in its root.

The main contribution of this paper is a characterization of the
regular tree languages that can be defined by a boolean query of
$\fotwo$.  It is, however, natural to also ask for a characterization
of unary queries.  For instance, the first unary query below can be
defined in $\fotwo$, but the second one cannot:
\begin{enumerate}[$\bullet$]
\item Some ancestor of the  selected node has label $a$,
      i.e.~$\tlf^{-1} a$.
\item Some child of the selected node has label $a$.
\end{enumerate}
In general, a regular unary query can be given e.g.~as a formula of
monadic-second order logic with one free variable.  Note that although
the second unary query cannot be defined, the tree language ``some
child of the root has label $a$'' can be defined, by the formula
\[
  \tlef (a \ \land \tlf^{-1} \true \land \neg \tlf^{-1} \tlf^{-1} \true)\ .
\]
This suggests that characterizing unary queries is a nonobvious
problem, which we leave as future work.

\subsection{Antichain composition  principle}
A problem with $\fotwo$ is that it is not closed under
``composition''. We illustrate this problem, together with a
workaround, for words; then we show the result for trees.

Consider the word languages $aa$ and $(a+b)^*$. Both are definable in
$\wordfotwo$,  and even  only using $\tlf$, but the language $(a+b)^*aa(a+b)^*$ is not. We claim
however, that the concatenation of two definable languages is also
definable if the place in the word where they meet can be uniquely
determined in $\wordfotwo$:
\begin{lem}[Composition for words]\label{lemma:conc-words}
  Let $L,K$ be two word languages definable in $\wordfotwo$ and let
  $\varphi$ be a $\wordfotwo$ formula with the semantic property that
  in every word, $\varphi$ holds in at most one word position.  The
  following word language is also definable in $\wordfotwo$:
  \[
    \set{a_1 \ldots a_n : a_1 \cdots a_i  \in L,\ a_{i+1} \cdots a_n
      \in K,\ \mbox{and $\varphi$ holds in $a_1 \cdots a_n$ at
        position $i+1$}}
  \]
\end{lem}
\begin{proof}
  We use  relativization.  We define $\psi_1$ by taking the
  formula defining $L$, and replacing each subformula $\psi$ by $\psi
  \land \tlf \varphi$. Likewise, we define $\psi_2$ by taking the
  formula defining $K$, and replacing each subformula $\psi$ by $\psi
  \land (\varphi \lor \tlf^{-1} \varphi)$. The formula for the
  language in the lemma is then $\psi_1 \land \tlf (\varphi \land
  \psi_2)$.
\end{proof}

For trees, the situation is more complicated. First of all, there are
two notions of composition: concatenation $s+t$ for forests and
composition $pq$ for contexts.  We are interested in generalizing
Lemma~\ref{lemma:conc-words} to composition of contexts.  In our
generalization though, we may need to substitute many trees
simultaneously.  This leads to a slightly less appealing definition,
which follows.

A formula is called \emph{antichain} if in every tree, the set of
nodes where it holds forms an antichain, i.e.~a set (not necessarily
maximal) of nodes pairwise incomparable with respect to the descendant
relation.  This is a semantic property, and may not be apparent just
by looking at the syntax of the formula. For instance, the first two
formulas below are antichain, while the third is not:
\begin{enumerate}[$\bullet$]
\item The node is a leaf: $\neg \tlef \true$.
\item The node is a minimal occurrence of $b$: $b \land \neg \tlf^{-1}
  b$.
\item The node has label $b$.
\end{enumerate}

Using antichain formulas, we define our notion of concatenation.  The
ingredients are:
\begin{enumerate}[$\bullet$]
\item An antichain formula $\varphi$.
\item Disjoint tree languages $L_1,\ldots,L_n$.
\item Leaf labels $a_1,\ldots,a_n$.
\end{enumerate}
Let $t$ be a tree.  We define the tree
\[
   t[(L_1,\varphi) \to a_1,\ldots,(L_n, \varphi) \to a_n]
\]
as follows. For each node $x$ of $t$ where the antichain formula
$\varphi$ holds, we determine the unique $i$ such the tree language
$L_i$ contains the subtree of $x$. If such an $i$ exists, we remove
the subtree of $x$ (including $x$), and replace $x$ by a leaf labeled
with~$a_i$.  Since $\varphi$ is antichain, this can be done
simultaneously for all $x$. Note that the formula $\varphi$ may depend
also on ancestors of $x$, while the languages $L_i$ only talk about
the subtree of $x$.

\begin{lem}[Antichain composition principle]
  Let $\varphi$, $L_1,\ldots,L_n$ and $a_1,\ldots,a_n$ be as above. If
  $L_1,\ldots,L_n$ are tree-definable, and $K$ is a tree-definable
  language, then so is
  \[
    \set{  t :  t [(L_1,\varphi) \to a_1,\ldots,(L_n, \varphi)
      \to a_n]  \in K} \ .
\]
\end{lem}
\begin{proof}
  This is proved by a relativization entirely analogous to the one
  used in Lemma~\ref{lemma:conc-words}.
\end{proof}

The point of this lemma is that the languages $L_i$ are taken out of
their context inside the tree $t$. For instance $L_i$ can say
something like: ``the root has label $a$ and a child with label $b$'',
\[
  L_i = \tlef (b \ \land\  \tlf^{-1} a\  \land\  \neg \tlf^{-1}
  \tlf^{-1} \true)\ ,
\]
while in general the property ``a node in the tree that has label $b$
and a child with label $b$'' cannot be expressed in $\fotwo$.

\section{Forest algebra}
\label{sec:forest-algebra}
To represent languages of trees, we will be using forest algebra. We
feel that using forest algebra instead of automata simplifies the
combinatorics used in our characterization.  Furthermore, when using
forest algebra, the key properties from Theorem~\ref{thm:main} can be
stated in terms of identities.

Here we only sketch out the definitions and basic properties; the
reader is referred to~\cite{forestalgebra} for more details. The
algebras described in~\cite{forestalgebra} differ slightly from those
used here---mainly in that we do not allow empty forests here---but
the results carry over into this setting.

A forest algebra is to a regular language of unranked trees as a
semigroup is to a regular language of words. Formally, a forest
algebra is an algebra with two sorts $(H,V)$, along with some
operations that satisfy a number axioms. While defining the operations
and axioms, we will illustrate them on an important example, called
the free forest algebra, where $H$ is the set of all nonempty forests,
and $V$ is the set of all, possibly empty, contexts.

The operations and axioms of forest algebra are presented below.
Elements of $H$ will be denoted by $h,g,f$ and elements of $V$ will be
denoted by $v,w,u$.
\begin{enumerate}[$\bullet$]
\item A composition operation $+$ on $H$.  This operation is required
      to be associative, i.e.~$h+(g+f)=(h+g)+f$ holds for all $f,g,h
      \in H$. This makes $H$ a semigroup, called the \emph{horizontal
        semigroup}, and justifies the notation $h+g+f$. In the free
      forest algebra, $+$ is forest concatenation. We do not require
      $H$ to contain a neutral element, e.g.~~there is no empty forest
      in the free forest algebra.
\item A composition operation $\cdot$ on $V$. Again, this is required
      to be associative. We omit the $\cdot$ symbol, writing $vw$
      instead of $v \cdot w$, for $v,w \in V$. Furthermore, we require
      there to be a neutral element $\hole \in V$, i.e.~an element
      satisfying $v\cdot\hole=\hole\cdot v=v$ for all $v \in V$. In
      particular, $V$ is a monoid, called the \emph{vertical monoid}.
      In the free forest algebra, $\cdot$ is context composition,
      while $\hole$ is the empty context.
\item An insertion operation $V \to H \to H$. The result of this
      insertion is denoted by $vh \in H$.  The empty context acts as
      the identity of this operation, i.e.~$\hole h=h$. The insertion
      operation must be a left action, i.e.~it must satisfy
      $(vw)h=v(wh)$ for $v,w \in V$ and $h \in H$, which justifies the
      notation $vwh$. In the free forest algebra, the left action is
      substituting a forest into a context. There is an faithfulness
      requirement:  distinct contexts $v,w \in V$ must induce
      different functions.
\item An operation $\mathit{left}: H \times V \to V$. This operation must
      satisfy $\mathit{left}(h,v)g = h+ vg$ for $v \in V$ and  $g,h \in H$. Thanks to this
      axiom, we can without ambiguity write $h+v$ to denote the
      element $\mathit{left}(h,v)$.  In the free forest algebra, $h+v$ is the
      context obtained from $v$ by prepending the forest $h$ (next to
      the root, not the hole). In a similar way we define $v+ h$, in
      terms of an operation $\mathit{right}$. 
\end{enumerate}

As demonstrated above, the free forest algebra is a forest algebra.
Clearly the free algebra depends on the leaf labels $A$ and inner node
labels $B$ (and only on these); once these are given, the free algebra is denoted by
$\freef A B$. When describing a forest algebra, we usually only give
names to the carrier sets $H$ and $V$, leaving the operations
implicit.

Let $(H,V)$ and $(G,W)$ be two forest algebras.  A \emph{forest
  algebra morphism}
\[
\alpha : (H,V) \to (G,W)
\]
is a pair of functions
\[
\alpha=(\alpha_H,\alpha_V) \qquad
  \alpha_H : H \to G \qquad \alpha_V : V \to W
\]
that preserve all operations in the signature, namely, composition $+$
in $H$,  composition $\cdot$  in $V$, insertion, and the $\mathit{left},
\mathit{right}$ operations.  For instance, preserving insertion is:
\[
   \alpha_H(vh) = \alpha_V(v)(\alpha_H(h))\ .
\]
To avoid clutter, we omit the subscripts, and write $\alpha(h)$
instead of $\alpha_H(h)$, likewise for~$v$.  

If $\alpha$ is a morphism, then the \emph{type under $\alpha$} of a
forest $t$ is simply the value $\alpha( t)$. Whenever the morphism
$\alpha$ is clear from the context, we omit the qualifier ``under
$\alpha$''. 

In this paper, a forest algebra will either be a free forest algebra,
or a finite forest algebra. In the first case, elements of the first
sort will be called \emph{forests} and denoted by $s,t$, while
elements of the second sort will be called \emph{contexts}, and
denoted by $p,q$. In the second case, of a finite forest algebra,  elements of the first
sort will be called \emph{forest types} and denoted by $f,g,h$, while
elements of the second sort will be called \emph{context types}, and
denoted by $u,v,w$.

\subsection{Equivalence with regular languages}
In this section we show that forest algebras provide an equivalent
description of regular tree languages.  Although this has already been
shown in~\cite{forestalgebra}, we present the proof here for two
reasons. First, our definition is slightly different from the one
in~\cite{forestalgebra}, where a neutral element was required in $H$.
Second, the notion of semigroup automaton used in the equivalence will
be used later on in the paper.

The point of forest algebras is to recognize forest languages.  Let
$L$ be a set of forests over labels $(A,B)$ and let $(H,V)$ be a
finite forest algebra.  We say a morphism
\[
  \alpha : \freef A B \to (H,V)
\]
\emph{recognizes} a forest language $L$ if membership $ t \in L$
depends only on the value $\alpha(t)$. In this case, we also say that
the algebra $(H,V)$ recognizes the language $L$.  Note that this definition is
for languages of forests, and not languages of trees, as in the logic
$\fotwo$. We will deal with this discrepancy in
Section~\ref{sec:tree-defin-forest}.

Below we show that forest algebras recognize exactly the regular
forest languages. What is a regular forest language? The definition
used here, of a semigroup automaton, is chosen so that the translation
to forest algebra is easiest.  A \emph{semigroup automaton} is a type
of bottom-up finite automaton that can be used to recognize tree and
forest languages. Let $(A,B)$ be an alphabet. A semigroup automaton
$\Aa$ over $(A,B)$ is defined by a finite semigroup $H$, whose
operation is denoted additively by $+$, along with two mappings (which
describe the initial states and transitions, respectively):
\[
\beta_A : A \to H \qquad
\beta_B : B \to H^H
\]
The purpose of the automaton is to uniquely associate a type
$\beta( t) \in H$ to every forest $ t$.  This is done using
the following rules:
\[\eqalign{
  \beta(a) &= \beta_A(a) \cr
  \beta(s_1 + \cdots + s_n) &= \beta(s_1) + \cdots + \beta(s_n)\cr
  \beta(b t)  &= \beta_B(b) (\beta( t))\ .\cr
  }
\]
Recall that in the last line above, $bt$ is a tree that has $b$ in the
root and the forest $t$ below.

An automaton recognizes a forest language $L$ if
membership $t \in L$ depends only on the value $\beta$. In other
words, one can choose a set of accepting elements $F \subseteq H$ such
that a forest $t$ belongs to $L$ if and only the value $\beta(t)$
belongs to $F$. The definition can be modified for  recognizing tree languages  by
requiring the equivalence $t \in L \iff \beta(t) \in F$ to hold only
for trees. Note that even when recognizing a tree language, a
semigroup automaton is still obliged to assign a value from $H$ to every forest.

It is not
difficult to show that this definition is equivalent to other existing
automata models for unranked trees, although there may be an
exponential blowup when translating to semigroup automata.

\begin{thm}\label{thm:equiv-with-regul}
  A forest language is regular if and only if it is recognized by a
  finite forest algebra. 
\end{thm}
\begin{proof}
Once we have a semigroup automaton, we can extend the mapping $\beta$
so that  contexts also get values, namely values in $H^H$. A context $p$ is assigned the
following mapping $\beta(p) \in H^H$:
\[
 h \mapsto \beta(p  t)\ ,
\]
where $ t$ is some forest with $\beta( t)=h$ (the choice of $ t$ does
not change this value). It is easy to see that the mapping $\beta$
(when seen as a mapping on both forests and contexts) is a forest
algebra morphism
  \[
    \beta : \freef A B \to (H,H^H)\ .
  \]
  This shows the harder direction in the proof of
  Theorem~\ref{thm:equiv-with-regul}. The other direction, from a
  forest algebra to a semigroup automaton, is immediate.
\end{proof}

\subsection{Syntactic algebra}

The syntactic forest algebra of a forest language $L$ is a
canonical forest algebra that recognizes the language. It is defined
using the following Myhill-Nerode equivalence over forests and
contexts. Two forests $ s, t$ are considered equivalent if for every
context $p$, either both or  neither $ps$ nor $pt$ belongs to $L$.
 Two contexts $p,q$ are considered equivalent if for every forest
  $ t$, the forests $p  t$ and $q  t$ are equivalent in
  the above sense.
  
  It turns out that the above defined equivalences are a congruence
  with respect to all operations in a forest algebra; therefore a
  quotient forest algebra can be defined, where elements of $H$ are
  equivalence classes of forests, and elements of $V$ are equivalence
  classes of contexts. This quotient forest algebra is called the
  \emph{syntactic forest algebra} of $L$. The \emph{syntactic
    morphism} is the morphism that assigns to each forest
  (resp.~context) its equivalence class. The syntactic morphism
  recognizes $L$, furthermore it is optimal in the sense that the
  syntactic morphism factors through any morphism recognizing $L$,
  i.e.~if $\beta$ is a morphism recognizing $L$, and $\alpha$ is the
  syntactic morphism of $L$, then there is a (unique) morphism $\gamma$ with
  $\alpha = \gamma \circ \beta$. In particular, the syntactic forest
  algebra is a morphic image of any forest algebra recognizing $L$,
  and a language has a finite syntactic algebra if and only if it is
  regular.

\subsection{Green's relations for trees}\label{sec:greens-relat-trees}
Fix a forest algebra $(H,V)$.  In this section we introduce two
preorders on $V$ and $H$ that will be used in the paper.

We say that context type $v \in V$ is \emph{reachable} from a context
type $w \in V$ if $v=wu$ holds for some context type $u \in V$.  A
\emph{context component} is a maximal set of mutually reachable
context types. Stated differently, two context types $v,w$ are in the
same context component if the ideals $vV$ and $wV$ are equal. Since
reachability is transitive and reflexive, it induces an order (not
necessarily linear) on context components.

\newcommand{\rreach}{\sqsubseteq} \newcommand{\reachn}{\sqsubsetneq} 
\newcommand{\reach}{\sqsupseteq} \newcommand{\rreachn}{\sqsupsetneq}

We say a forest type $g \in H$ is \emph{reachable} from a forest type
$h \in H$ if $g=uh$ holds for some context type $u \in V$. A
\emph{forest component} is a maximal set of mutually reachable
forests.  Stated differently, two forest types $g,h$ are in the same
forest component if the ideals $Vg$ and $Vh$ are equal. As for context
types, forest components are ordered by reachability.  Note that $g+h$
is reachable from $h$, since we can take the context type $u$ to be $g
+ \hole$.

These two preorders are related to Green's relations used in semigroup
theory. Actually, reachability on contexts simply is the $\mathcal
R$-order on the semigroup $V$. The reachability relation on $H$ is not
one of Green's relations, since its definition involves the two sorts
$H$ and $V$ in the forest algebra.

\section{Tree-Definable vs Forest-Definable}
\label{sec:tree-defin-forest}
A tree language $L$ is \emph{tree-definable} if there is a formula of
$\fotwo$ that is true exactly (in the root of) trees in $L$.  In this
paper, it will sometimes be convenient to talk about $\fotwo$ formulas
defining properties of forests (and not only trees). We say a forest
language $L$ is \emph{forest-definable} if $L$ is a boolean
combination of languages of the form ``some tree in the forest
satisfies $\varphi$'', with $\varphi$ a formula of $\fotwo$.  Such a
boolean combination will be called a \emph{forest formula}.  For
instance, the following property of a forest $t_1 + \cdots + t_n$ is
forest-definable: all trees $t_1,\ldots,t_n$ contain a leaf with label
$a$, and at least one of these trees has root label $b$.  Any nonempty
tree language violates the following property, which is true for
forest-definable languages:
 \[
  t+ t \in L \quad \miff \quad t \in L \ ,
\]
for the simple reason that $t+t$ is not a tree. Therefore no nonempty
tree language is forest-definable. For the same  reason, no nonempty
forest-definable language is tree-definable. 

In this paper, we will present a decidable characterization for
forest-definable languages. Thanks to the following result, this will
also give us a decidable characterization of tree-definable languages.
\begin{prop}\label{prop:from-tree-to-forest}
  Let $L$ be a tree language over $(A,B)$. The following conditions
  are equivalent:
  \begin{enumerate}[$\bullet$]
  \item $L$ is  tree-definable.
  \item For each inner node label $b \in B$, the forest language
    $\set{ t : b t \in L}$ is forest-definable.
  \end{enumerate}
\end{prop}
\begin{proof}
  
  We begin by showing that the first property implies the second. Assume then that $L$ is
  tree-definable, and fix some $b \in B$. We need to show that the
  forest language $\set{t : b t \in L}$ is forest definable. 
  
  Let $P$ be the set of contexts of the form $p=b(\hole + t)$, where
  $t$ is a forest. Consider the following equivalence relation on
  trees:
  \[
    s \sim t \qquad \miff \qquad ps \in L \iff pt \in L \quad\mbox{ holds for all
      }p \in P\ .
  \]
  This equivalence relation has only finitely many classes, since it
  is coarser than the Myhill-Nerode equivalence relation used in the
  definition of syntactic algebra.  Note that we would get the same
  equivalence relation by also considering contexts of the form $p=(s
  + \hole +t)$, since $\fotwo$ is invariant under reordering siblings.
  Furthermore, each of these equivalence classes is tree-definable,
  thanks to the following fact: if $p$ is a context and $L$ a
  tree-definable language then the set of trees $t$ with $pt \in L$ is
  tree-definable. The standard proof of this fact is omitted here.
  For any forest $t=t_1 + \cdots + t_n$, membership $bt \in L$ only
  depends on the equivalence classes under $\sim$ of the trees
  $t_1,\ldots , t_n$ that the constitute the forest $t$.  Since
  $\fotwo$ formulas are invariant under duplicating and reordering
  sibling trees, it is only the set of equivalence classes that
  counts, which can be described by a boolean combination of languages
  of the form required in forest-definable languages.

  We now do the bottom-up implication.  It suffices to show that if a
  forest language $L$ is forest-definable, then for any inner node
  label $b \in B$, the tree language $\set{ b t : t \in L}$ is
  tree-definable.  The key step is that if a tree language $K$ is
  tree-definable, then the following tree language:
    \[
      \tlx K =\set{b(t_1+ \cdots + t_n) :  b \in B, \exists i.\ t_i \in K}
    \]
    is also tree-definable. Once we demonstrate how to write a formula
    for $\tlx K$, the formula tree-defining $b L$ can be obtained from
    the formula forest-defining $L$.
    
    Note that definability of the language $\tlx K$ does not mean we
    can add the child operator to the logic. This is because $\tlx K$
    uses the child only at a fixed depth. For instance, the property
    ``some node at depth 4 has the same label as its parent'' is
    tree-definable, contrary to the property ``some node has the same
    label as its parent''.
    
    The formula for $\tlx K$ can be obtained from the antichain
    composition principle, but we do a direct construction here. Let
    $\varphi$ be the formula defining $K$. We define $\hat \varphi$ to
    be the formula obtained from $\varphi$ by replacing every
    subformula $\psi$ by $ \psi \land \tlf^{-1} \true$. This way,
    quantification in $\hat \varphi$ is relativized to non-root nodes.
    Finally, the formula for $\tlx K$ is
    \[
      \tlef \big((\tlf^{-1} \true) \land (\neg \tlf^{-1} \tlf^{-1} \true)\land
      \hat \varphi\big)\ .
    \]
    The above formula nondeterministically  picks a  successor $x$ of the
    root, and then tests if $\hat \varphi$ holds in $x$. Since  $\hat
    \varphi$ is relativized to non-root nodes, evaluation of $\hat
    \varphi$ will never leave the subtree of $x$.
      \end{proof}

    \section{The identities and the main result}
In this section we state our main result, the decidable
characterization of the logic $\fotwo$.

The characterization uses a relation~$\mainle$ over contexts in a
forest algebra. The idea is that $u \mainle w$ holds if the context $u$ can be
obtained from the context $w$ by removing forests that are siblings of the main
path (recall that the main path contains ancestors of the hole). Let
$(H,V)$ be a forest algebra. For $u,w \in V$, we write $u \mainle w$
if $u, w$ can be decomposed as
\[
 u = v_0v_1 \cdots v_n \qquad   w=v_0(h_1 + v_1) \cdots (h_n +v_n) 
\]
for some $v_0,\ldots,v_n \in V$ and $h_1,\ldots,h_n \in H$.  The
reason why we have $v_0$ above, and not $h_0 + v_0$, is that a context
type can be empty, but there is no empty forest type.  The following
lemma shows that the relation $\mainle$ can be calculated in
polynomial time using a least fixpoint algorithm:
\begin{lem}\label{lemma:fixpoint}
The relation $\mainle$ is the least relation $R \subseteq V\times V$
such that:
\[\eqalign{
(v,v),(v,v+h),(v,h+v) \in R &\qquad \mbox{for }v \in V,\ h \in H\cr
(v,v'),(w,w') \in R \Rightarrow  (vw,v'w') \in R& \qquad \mbox{for
}v,v',w,w' \in V\ .}
\]
\end{lem}
\begin{proof}
  The implication from $(v,w) \in R$ to $v \mainle w$ is proved by
  induction on the number of steps in the derivation. The converse
  implication is proved by induction on $n$ in the definition of
  $\mainle$.
\end{proof}

The relation $\mainle$ is transitive in some forest algebras,
including all free forest algebras. However, in general it need not be
transitive, as illustrated by the following example.  Let the leaf
alphabet $A$ be $\set{a_1,a_2}$ and let the inner node alphabet $B$ be
$\set b$. Consider the forest language $L$: ``the forest does not
contain both labels $a_1$ and $a_2$ at the same time, and every node
with label $b$ has a sibling with label $a_1$ or $a_2$''.  Let
$\alpha$ be the syntactic morphism of this language.  Consider the
following four contexts: \medskip

  \begin{center}
    \includegraphics[scale=0.65]{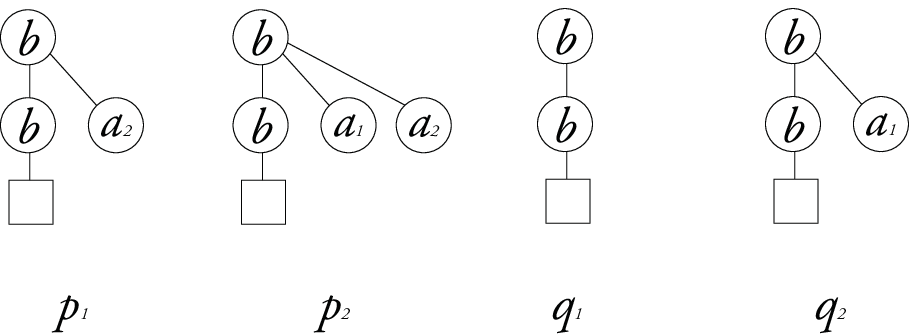}
  \end{center}
  Clearly we have $\alpha(p_1) \mainle \alpha(p_2)$ and $\alpha(q_1)
  \mainle \alpha(q_2)$. We claim that $\alpha(p_2)=\alpha(q_1)$.
  Indeed, both contexts are ``error'' contexts, i.e.~for any context
  $r$ and forest $t$ we have $rp_2t,rq_1t \not \in L$. Therefore, if
  $\mainle$ were a transitive relation, we would have $\alpha(p_1)
  \mainle \alpha(q_2)$. This, however, cannot hold, since otherwise we
  could construct a tree in $L$ with both $a_1$ and $a_2$ labels.

We are now ready to state the main theorem of this paper:
\begin{thm}\label{thm:main}
  A language is forest-definable in $\fotwo$ if and only if its
  syntactic algebra satisfies the following identities:
  \begin{equation}
    \label{eq:apcom}
    h+h=h \qquad g+h = h+g
  \end{equation}
  \begin{equation}
    \label{eq:da}
    (vw)^\omega = (vw)^\omega w (vw)^\omega\ .
  \end{equation}
  \begin{equation}
    \label{eq:special}
    \begin{array}{cc}
      (u_1 w_1 )^\omega (u_2 w_2 )^\omega = 
      (u_1 w_1 )^\omega u_1 w_2  (u_2 w_2 )^\omega  
      \qquad \mbox{if $u_1 \mainle u_2, w_1 \mainle w_2$}\ .
    \end{array}
  \end{equation}
\end{thm}

In the identities above, all variables are quantified universally.
The identities in~(\ref{eq:apcom}) say that children can be duplicated
and reordered. This corresponds to bisimulation invariance in the
following way: a forest  language is bisimulation invariant if and only
if its syntactic forest algebra satisfies~(\ref{eq:apcom}).
The identity~(\ref{eq:da}) says that the vertical monoid belongs to
the variety DA (although the commonly used identity is different).
Only the last identity is new.

The exponent $\omega$ in properties~(\ref{eq:da})
and~(\ref{eq:special}) stands for ``for almost all $n$''. In
particular, identity~(\ref{eq:da}) should be read as:
\[
\exists m \forall n \ge m \qquad    (vw)^n = (vw)^n w (vw)^n\ .
\]
Usually in semigroup theory, $\omega$ stands for ``least idempotent
power'', but the above definition is equivalent for aperiodic monoids,
which is the case here, thanks to~(\ref{eq:da}).

An important corollary of the above theorem is that definability in
$\fotwo$ is decidable:
\begin{cor}
  It is decidable if a forest (resp.~tree) language is
  forest-definable (resp.~tree-definable) in $\fotwo$. The algorithm
  runs in polynomial time if the input is given as a forest algebra.
\end{cor}
\begin{proof}
  To determine if a language is tree-definable, we calculate the
  languages $\set{ t : b t \in L}$ and reduce to the
  characterization of forest-definable language thanks to
  Proposition~\ref{prop:from-tree-to-forest}.  Therefore, we focus on
 deciding if a language is forest-definable.

  We begin by finding the syntactic forest algebra.  The syntactic
  forest algebra can be effectively calculated based on any
  representation of the tree language, be it a tree automaton, or a
  formula of some rich logic, such as MSO. In general, the syntactic
  forest algebra can be exponentially larger than a nondeterministic
  tree automaton, not to mention a formula of MSO.
  
  Once the syntactic forest algebra has been calculated, the
  properties~(\ref{eq:apcom}),~(\ref{eq:da}) and~(\ref{eq:special})
  can be verified in polynomial time (with respect to the algebra).
  The relation $\mainle$ over $V$ can be computed in polynomial time
  thanks to Lemma~\ref{lemma:fixpoint}.  The exponent $\omega$ is not
  a problem. Indeed, a consequence of (\ref{eq:da}) is that $V$ is
  \emph{aperiodic}, i.e.~the identity $v^\omega = v^\omega v$ holds
  for all context types $v$. In particular, it is enough to test for
  $\omega = |V|$.
\end{proof}

The rest of this paper is devoted to showing Theorem~\ref{thm:main}.
The ``only if'' implication in the above theorem is proved in
Section~\ref{sec:correctness} using a simple induction on formula
size.  The difficult part is the proof of the ``if'' implication,
which is found in Section~\ref{sec:completeness}.

In the following fact, we show that property~(\ref{eq:special}) in
Theorem~\ref{thm:main} is not redundant. In a similar way one can
prove that neither~(\ref{eq:apcom}) nor~(\ref{eq:da}) is redundant.
\begin{lem}
  There exists a forest algebra satisfying properties~(\ref{eq:apcom})
  and~(\ref{eq:da}) but not~(\ref{eq:special}).
\end{lem}
\proof
  Let the leaf alphabet $A$ be $\set{a_1,a_2}$ and let the inner node
  alphabet $B$ be $\set b$. Consider the following language: ``if a
  node has a child with label $a_1$, then it has an ancestor with a
  child with label~$a_2$''.
  The syntactic forest algebra of this language satisfies
  properties~(\ref{eq:apcom}) and (\ref{eq:da}); but it does not
  satisfy~(\ref{eq:special}), since for all $n \in \Nat$ we have
  $$
    (bb)^n((b+a_2)(b+a_1))^n a_2  \in L \qquad
    (bb)^n b(b+a_1)((b+a_2)(b+a_1))^n a_2 \not \in L \ .\eqno{\qEd}
  $$

\section{Correctness}\label{sec:correctness}
In this section we show that any language forest-definable in $\fotwo$
satisfies the identities from Theorem~\ref{thm:main}.  For each of
these identities we show that any formula of $\fotwo$ must, informally
speaking, confuse the two trees described by the opposing sides of the
identity. To show this confusion, we use an Ehrenfeucht-Fra\"iss\'e
game.  The plan of this section is as follows. First, in
Section~\ref{sec:ehrenf-fraisse-games}, we define the
Ehrenfeucht-Fra\"iss\'e that characterizes $\fotwo$. Next, in
Section~\ref{sec:morphic-images}, we use the game to show that
languages defined in $\fotwo$ are closed under morphic
preimages. Finally, in Section~\ref{sec:corr-ident} we show that any
language forest-definable in $\fotwo$ satisfies the identities from
Theorem~\ref{thm:main}.

\subsection{Ehrenfeucht-Fra\"isse Game}
\label{sec:ehrenf-fraisse-games}
In this section, we define an Ehrenfeucht-Fra\"iss\'e game that
characterizes the logic $\fotwo$.

The game is played on two forests $s_0$ and $s_1$, with two
distinguished nodes, $x_0$ in $s_0$ and $x_1$ in $s_1$. A
configuration of the game is therefore a four-tuple
$(x_0,x_1,s_0,s_1)$. Finally, the game has a parameter $n \in \Nat$,
which is called the number of \emph{rounds}.  The game is played by
two players, Duplicator and Spoiler.  The idea is that Duplicator
claims that the same formulas of size at most $n$ hold in $x_0$ and
$x_1$.

The game is played as follows. Assume that there are $n \ge 0$ rounds
left. If the labels of $x_0$, $x_1$ are different, then Spoiler wins
the game immediately, and no further rounds are played. If the labels
are the same, and $n=0$, then Duplicator wins the game, and no further
rounds are played. Finally, if the labels are the same and $n > 0$, a
new round is played as follows.

First, Spoiler chooses one of the two nodes $x_0,x_1$, i.e.~he chooses
an index $i \in \set{0,1}$. The idea is that Spoiler thinks that the
node $x_i$ has some property that the other node $x_{1-i}$ does not
have.  He then chooses to make either a descendant move (in this case,
Spoiler thinks that $x_i$ has a descendant unlike all descendants of
$x_{1-i}$) or an ancestor move (Spoiler thinks that $x_i$ has an
ancestor unlike all ancestors of $x_{1-i}$) . If Spoiler chooses a
descendant (respectively, ancestor) move, then he must choose a proper
descendant (respectively, proper ancestor) $y_i$ of $x_i$ in the
forest $s_i$. To this, Duplicator must respond by choosing a proper
descendant (respectively, proper ancestor) $y_{1-i}$ of $x_{1-i}$ in
the other forest $s_{1-i}$.  The idea is that Duplicator thinks that
$y_{1-i}$ is similar to $y_i$, at least as far as the remaining $n-1$
rounds are concerned. Formally, the new configuration becomes
$(y_0,y_1,s_0,s_1)$ and the game continues with $n-1$ rounds left.

We also define how the $n$-round game is played on two forests
$s_0,s_1$ in case when the nodes $x_0,x_1$ are not specified. In this
case, there is a special introductory round, where Spoiler chooses $i
\in \set{0,1}$ and a root node $x_i$ in $s_i$; Duplicator responds
with a root node $x_{1-i}$ in the other forest. Then the standard
$n$-round game continues from this configuration.

\begin{prop}\label{prop:ef-char}
  A forest language is forest-definable in $\fotwo$ if and only if for
  some~$n$, Spoiler wins the $n$-round game for any pair of forests
  $s_0 \in L$ and $s_1 \not \in L$.
\end{prop}
\begin{proof}
  The proof is standard, and omitted here. The idea is that $n$ is the
  nesting depth of the formulas used to forest-define $L$. The nesting
  depth counts the maximal nesting of $\tlef$ and $\tlf^{-1}$ in a
  formula, while boolean operations are for free.
\end{proof}

\subsection{Morphic images}
\label{sec:morphic-images}
In this section, we show that languages forest-definable in $\fotwo$
are closed under morphic preimages.  Actually, we show a slightly more
general result. The more general setting will be used in
Section~\ref{sec:empty-forests}, where we show that our
characterization also works for a different model of forest algebra,
where empty forests are allowed.

We first describe the more general setting.  The generalization is
twofold. First, we allow empty forests. Second$^*$\footnote{$*$ It
  turns out that in forest algebra, the first generalization entails
  the second.}, we consider forests over a single alphabet (unlike the
two-sorted alphabet $A,B$ considered before, with $A$ allowed only in
leaves and $B$ allowed only in inner nodes).  The new type of forests
will be called \emph{one-sorted forests}, to distinguish them from the
\emph{two-sorted forests} considered before.  The one-sorted forests
are more general in the following sense: the two-sorted forests over
an alphabet $(A,B)$ are a subset of the one-sorted forests over the
alphabet $A \cup B$.  Of course, the difference is not that big: the
one-sorted forests over $A$ are the two-sorted forests over $(A,A)$,
plus the empty forest.We also have an analogous concept of
\emph{one-sorted contexts}.  A \emph{one-sorted morphism}, with
\emph{source alphabet $A$} and \emph{target alphabet $B$} is given by
a function that assigns to each letter of $A$ a one-sorted context,
possibly empty, over $B$. A one-sorted morphism uniquely extends to
one-sorted forests and one-sorted contexts. To avoid confusion, in
this section we use the name \emph{two-sorted morphism} for the
morphisms introduced previously in the paper.

\begin{thm}\label{thm:ef}
  Let $\alpha$ be a one-sorted morphism. If a forest language $L$ over
  the target alphabet $B$ is forest-definable in $\fotwo$, then so is
  its inverse image $\alpha^{-1}(L)$.
\end{thm}

The version of this theorem for two-sorted morphisms is a special case
of the one-sorted version, since every for two-sorted morphism there
is a one-sorted morphism that gives the same results over all legal
two-sorted forests.

To show this theorem, we will use the Ehrenfeucht-Fra\"iss\'e game. We
fix the forest-language $L$ and the (one-sorted) morphism $\alpha$
from the theorem for the rest of this section. Let $n$ be the number
of rounds obtained by applying Proposition~\ref{prop:ef-char} to the
forest $L$ in the statement of the theorem.  By invoking
Proposition~\ref{prop:ef-char} a second time, to establish that the
inverse image $\alpha^{-1}(L)$ is forest-definable in $\fotwo$, it suffices
to show that Spoiler can win the $n$-round game over any two
preimages, one taken from the preimage $\alpha^{-1}(L)$, and the other
taken from its complement. The proof will be by showing how a strategy
of Duplicator over the preimage can be lifted to a strategy over the
image, as stated in the following proposition.
\begin{prop}\label{prop:ef-transfer}
  If Duplicator wins the $n$-round game over $s_0,s_1$, then
  Duplicator also wins the $n$-round game over
  $\alpha(s_0),\alpha(s_1)$.
\end{prop}

To prove this transfer of strategies, we will be switching back and
forth between the Ehrenfeucht-Fra\"iss\'e games on $s_0,s_1$ and on
$\alpha(s_0),\alpha(s_1)$. To avoid confusion, we use the name
\emph{preimage game} for the former and we use the name \emph{image
  game} for the latter.  We will be comparing configurations of the two
games in the following way.  Every node $x$ in a morphic image
$\alpha(s)$ can be uniquely identified by two pieces of information:
its \emph{preimage} $\bar x$, which is a node in the preimage forest
$s$, and its \emph{offset}, which is a node of the context assigned by
$\alpha$ to the label in $\bar x$. These concepts are illustrated below, in
an example where both the source and target alphabets are $\set{a,b}$,
and the one-sorted morphism is defined by $\alpha(a)=a(\hole + b)$ and
$\alpha(b)=\hole$.
  \begin{center}
    \includegraphics[scale=0.65]{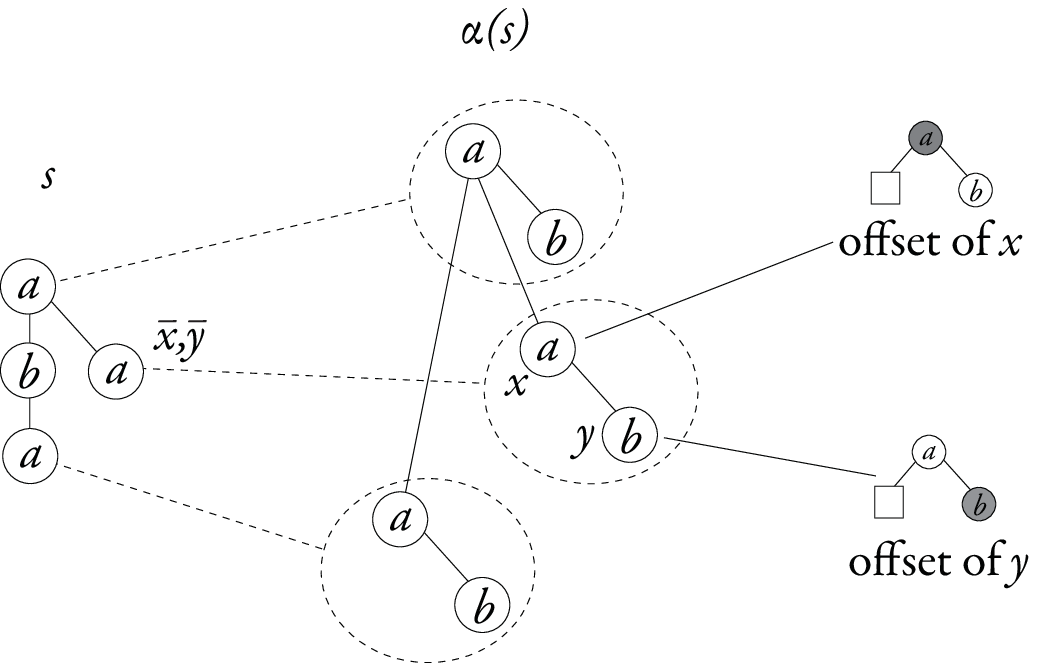}
  \end{center}
  Note that some nodes in the preimage forest $s$ are not the preimage
  of any node in $\alpha(s)$, these are the nodes whose labels are mapped
  to an empty context by $\alpha$.
  
  Armed with the definitions of offset and preimage, we now prove the
  strategy transfer from Proposition~\ref{prop:ef-transfer}. We only
  give the main invariant, which is described below. The missing part
  of the proof, for the introductory round of the game where the root
  nodes are chosen, is done in a similar way.

  \begin{lem}
    Let $m \le n$.  Let $x_0,x_1$ be nodes with the same offset 
    such that $\bar x_0, \bar x_1$ have the same label. If Duplicator
    can win the $n$-round preimage game in configuration $(\bar
    x_0,\bar x_1, s_0,s_1)$, then he can also win the $m$-round image
    game in configuration $(x_0,x_1,\alpha(s_0),\alpha(s_1))$.
  \end{lem}
\proof
  The proof is by induction on $n$. Consider first the case of
  $n=0$. By assumption on the preimage game, the nodes $\bar x_0$ and
  $\bar x_1$ have the same labels in $s_0,s_1$. Since the two nodes
  $x_0,x_1$ have the same offsets, they must also have the same labels
  in the images $\alpha(s_0),\alpha(s_1)$, and therefore Duplicator wins.

Consider now the induction step. We only do the case when Spoiler
chooses a descendant move, the ancestor move is done the same
way. Assume then that Spoiler chooses $x_i$ and indicates a proper
descendant $y_i$ of $x_i$ in $\alpha(s_i)$.  How should Duplicator respond?
There are two possible cases:
\begin{enumerate}[$\bullet$]
\item The preimage $\bar y_i$ is a proper descendant of $\bar x_i$. We
  now go to the preimage game, and make Spoiler play a descendant move
  where he chooses $\bar y_i$. By assumption on Duplicator winning the
  preimage game, there is a proper descendant of $\bar x_{1-i}$, call
  it $\bar y_{1-i}$, such that Duplicator wins the $(n-1)$-round
  preimage game from configuration $(\bar y_0,\bar y_1, s_0,s_1)$. In
  particular, the nodes $\bar y_0,\bar y_1$ have the same labels in
  the preimage, and therefore the same possible offsets in the
  image. Therefore, there exists a node $y_{1-i}$ in $\alpha(s_{1-i})$ such
  that its preimage is $\bar y_{1-i}$, and this node can be chosen to
  have the same offset as $\bar y_i$.  We now use the induction
  assumption to show that Duplicator wins the rest of the image game
  from configuration $(y_0,y_1,\alpha(s_0),\alpha(s_1))$.
\item If the preimage $\bar y_i$ is not a proper descendant of $\bar
  x_i$, then $\bar y_i = \bar x_i$ and the only difference between
  $y_i$ and $x_i$ is in the offset. Duplicator's response is to choose
  in the forest $s_{1-i}$ a node $y_{1-i}$ that has the same offset as
  $y_i$, and such that $\bar y_{1-i} = \bar x_{1-i}$. We then use the
  induction assumption to show that Duplicator wins the rest of the
  image game.\qed
\end{enumerate}

\subsection{Correctness of the identities}
\label{sec:corr-ident}
We are now ready to show the easier implication in
Theorem~\ref{thm:main}, namely that the syntactic forest algebra of a
language forest-definable in $\fotwo$ satisfies the three identities.
Validity of~(\ref{eq:apcom}) can easily be shown. We omit the proof of
(\ref{eq:da}) for two reasons: first, it is the same as in the word
case, see e.g.~\cite{therienwilkefo2}; and second, it follows along
similar lines as the proof of (\ref{eq:special}).

The rest of this section is devoted to showing the validity of
identity~(\ref{eq:special}). Let $L$ be a forest language
forest-definable in $\fotwo$. We need to show that the syntactic
algebra of $L$ satisfies identity~(\ref{eq:special}).  Recall that
elements of the syntactic algebra are equivalence classes of the
Myhill-Nerode equivalence relation.  Therefore, in order to show the
validity of (\ref{eq:special}), we have to show that for any formula
$\varphi$ of $\fotwo$, for all contexts $p_1 \mainle p_2$ and $q_1
\mainle q_2$, every context $p$ and every nonempty forest $ t$, for
almost all $n \in \Nat$ the formula $\varphi$ is true in some tree of
either both or neither of the forests
\begin{equation}
  \label{eq:non-diff}
      s_0 =     p(p_1 q_1 )^n (p_2 q_2 )^n  t \qquad
 s_1=  p(p_1 q_1 )^n p_1 q_2  (p_2 q_2 )^n   t \ .
\end{equation}
We will use the Ehrenfeucht-Fra\"iss\'e game, and show that Duplicator
can win the $n$-round game over the above two forests.
To keep notation simple, we assume the following \emph{simplifying
  assumptions} are met.
\begin{enumerate}[$\bullet$]
\item The context $p$ is a single node $b$ (in particular, $s_0$ and
  $s_1$ are trees).
\item The forest $ t$ is a single node $a$.
\item The contexts $p_1,p_2,q_1,q_2$ are
\[\eqalign{
    p_1 &= b_1 \cdots b_k\cr
    q_1 &= b_{k+1} \cdots b_m \cr
  }
  \qquad
  \eqalign{ 
    p_2 &= b_1(a_1 + \hole) \cdots b_k(a_k+ \hole)\cr
    q_2 &= b_{k+1}(a_{k+1} + \hole) \cdots b_m(a_m+ \hole)\cr
  }
\]
  for some $k < m$ and $b_1,\ldots,b_m \in B, \ a_1,\ldots,a_m \in A$.
\item The labels $a,a_1,\ldots,a_m,b,b_1,\ldots,b_m$ and $a$ are all
  distinct.
\end{enumerate}
The trees $s_0$ and $s_1$ are shown in Figure~\ref{fig:pumping}.  Why
can we make these simplifying assumptions? The reason is that the
general case follows from this special case by way of homomorphic
images. More specifically, consider the two forests $s_0, s_1$ in the
general case, as given in~(\ref{eq:non-diff}). We want to show that
Duplicator wins the $n$-round game over these two forests. The key
observation is that any two forests $s_0,s_1$ as
in~(\ref{eq:non-diff}) can obtained as homomorphic images
$s_0=\alpha(t_0)$ and $s_1=\alpha(t_1)$ from trees $t_0,t_1$ that
satisfy the simplifying assumption, for some (two-sorted) morphism
$\alpha$. As long as we know how Duplicator can win the game over the
simpler trees $t_0, t_1$, we can use
Proposition~\ref{prop:ef-transfer} to transfer this result to the
forests $s_0,s_1$.

\begin{figure}
  \centering
  \includegraphics[scale=0.65]{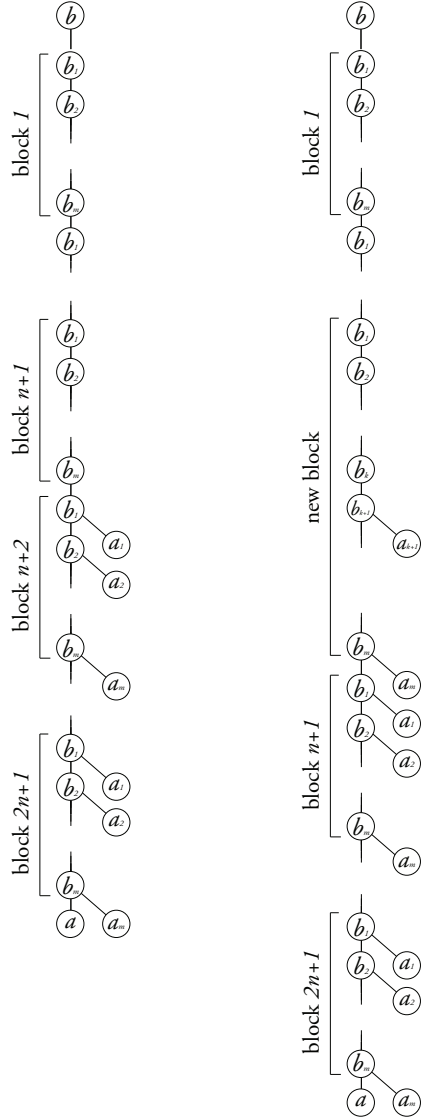}
  \caption{The trees $s_0$ and $s_1$}
  \label{fig:pumping}
\end{figure}

We now proceed to describe a winning strategy for Duplicator over
trees $s_0,s_1$ that satisfy the simplifying assumptions. We use the
term \emph{main path} for the ancestors of the node $a$. The
\emph{projection} of a node onto the main path is its closest ancestor
(not necessarily proper) that is on the main path.  For a node in
either $s_0$ or $s_1,$ the \emph{ancestor block count} (respectively,
\emph{descendant block count}) is the number of ancestors with label
$b_m$ (respectively, descendants with $b_1$) of the node's projection
onto the main path.  For $m \le n$, we say that two nodes $x_0,x_1$ in
the trees $s_0,s_1$ are \emph{$m$-similar} if their labels are the
same and moreover one of the conditions in the following invariant
holds:
\begin{enumerate}
      \item The trees $s_0,s_1$ agree on nodes in the subtrees of
        $y_0,y_1$; or
      \item The trees $s_0,s_1$ agree on nodes not in the subtrees of $y_0,y_1$; or
      \item The ancestor and descendant block counts of $x_0,x_1$ are
        both at least $m$.
\end{enumerate}

  \begin{lem}
    Let $m \le n$. If the nodes $x_0,x_1$ are $m$-similar, then
    Duplicator wins the $m$-round game from configuration
    $(x_0,x_1,s_0,s_1)$.
  \end{lem}
  \begin{proof}
    The proof is by induction on $m$. For the base case $m=0$ we use
    the assumption that the labels are the same. Consider now the
    induction step. We only do one case, when Spoiler chooses a
    descendant move to go from $x_1$ to a node $x'_1$ in the ``new
    block'' of $s_1$ (the new block is the context $p_1q_2$). This
    Spoiler move means that $x_0,x_1$ are $m$-similar for reason (2)
    or (3), since item (1) forbids a descendant of $x_1$ in the new
    block.  What is Duplicator's response?  Note that for all nodes in
    the new block, both the ancestor and descendant block counts are
    at least $n \ge m-1$.  Duplicator goes to any node $x'_0$ in the
    tree $s_0$ where the ancestor and descendant block counts are both
    at least $m-1$.  This must be possible, since either one of items
    (2) or (3) of the invariant was true for $x_0,x_1$. The rest of
    the game is played according the induction assumption, since
    $x'_0$ and $x'_1$ are $(m-1)$-similar.
  \end{proof}
  By taking $m = n$ in the above lemma, we get the desired
  result. This is because the two roots of $s_0,s_1$ have the same
  (empty) prefixes, thus they are $m$-similar, and must therefore
  satisfy the same formulas of size $m=n$.

\section{Completeness}
\label{sec:completeness}
This section is devoted to showing:
\begin{prop}\label{prop:completeness}
  Any forest language recognized by a forest algebra satisfying
  (\ref{eq:apcom}), (\ref{eq:da}) and (\ref{eq:special}) can be
  forest-defined.
\end{prop}
The above statement immediately implies the more difficult ``if'' part
of Theorem \ref{thm:main}. Indeed, if $L$ is recognized by an algebra
satisfying (\ref{eq:apcom}), (\ref{eq:da}) and (\ref{eq:special}),
then its syntactic algebra satisfies these identities. This is because
the syntactic algebra is a morphic image of any algebra recognizing
the language, and identities are preserved by morphic images.

Let $X \subseteq H$ be a set of forest types. We say a forest $ t$ is
\emph{$X$-trimmed} if the only subtrees of $ t$ that have a type in
$X$ are leaves. We say a tree language $L$ is tree-definable
\emph{modulo $X$} if there is a formula $\varphi$ such that
\[
  t \mbox{ satisfies } \varphi \qquad \miff \qquad t \in L
\]
holds for all $X$-trimmed trees (for other trees, $\varphi$ may
disagree with $L$). In a similar fashion, we define a forest language
that is forest-definable modulo $X$.

Instead of Proposition~\ref{prop:completeness}, we show the slightly
more general result below, which contains the induction parameters
that appear in the proof.
\begin{prop}\label{induction:prop}
  Let $ \alpha : \freef A B \to (H,V)$ be a morphism, with $(H,V)$
  satisfying identities~(\ref{eq:apcom}), (\ref{eq:da})
  and~(\ref{eq:special}).  Let $X \subseteq H$ be a set of forest
  types, and let $v \in V$ be a context type. For each forest type $h
  \in H$ the following forest language is forest-definable modulo $X$:
\begin{equation}
  \label{eq:main-ind}
  \set{  t :  v(\alpha( t))= h}\ .
\end{equation}
\end{prop}

For the rest of Section~\ref{sec:completeness}, we fix 
$\alpha : \freef A B \to (H,V)$, $h \in H$, $v \in V$ and $X \subseteq
H$ from Proposition~\ref{induction:prop}. 
Clearly Proposition~\ref{prop:completeness} follows from the above
result, taking $X=\emptyset$, $v$ to be the empty context type
$\hole$, and doing a disjunction over all forest types $h \in
\alpha(L)$.  The rest of Section~\ref{sec:completeness} is devoted to
a proof of Proposition~\ref{induction:prop}.  The proof is by
induction on four parameters:
\begin{enumerate}
\item The size of $H$, i.e.~the number of all forest types.
\item The size of $H \setminus X$, i.e.~the number of forest types that can be found outside leaves.
\item The size of $vV$, i.e.~the number of context types reachable
      from $v$.
\item The size of $B$, i.e. the number of inner node labels.
\end{enumerate}
The order of these parameters is important: first we try to minimize
$H$, then the other three parameters (the order for the other three is
not important). Note that the last parameter depends on the alphabet
$B$, and the notion ``modulo $X$'' depends on the morphism.

We say a morphism $\alpha$ into $(H,V)$ is \emph{leaf saturated} if
for every $h \in H$, there is a representative leaf label $a$ whose
type $\alpha(a)$ is $h$. In the rest of this section, we will only
consider such morphisms.  By adding leaf labels, any morphism can be
extended to one that is leaf saturated, without affecting the target
forest algebra.

We begin by outlining our proof strategy for
Proposition~\ref{induction:prop}. We will consider three possible
cases.  First, in Section~\ref{sec:not-all-labels-preserve}, we see
what happens when some inner node label $b \in B$ has the property
that $v$ cannot be reached from $v\alpha(b)$. Then, in
Section~\ref{sec:one-forest-comp}, we see what happens if $H \setminus
X$ intersects more than one forest component, i.e.~contains at least
two forest types that are not mutually reachable.  Finally, in
Section~\ref{sec:otherwise}, we show that if neither of the above
holds, then the formula $\varphi$ in Proposition~\ref{induction:prop}
can basically be replaced by either ``true'' or ``false''.

\subsection{For some inner node label $b \in B$, $v$ is not reachable
  from $v\alpha(b)$ }\label{sec:not-all-labels-preserve} We begin with
this case, which is the easiest of the three. The basic idea is that
we cut the forest into two parts, by looking at the first occurrence
of $b$ on each path, beginning at the root. Since after reading the
label $b$, the context type $v$ is no longer reachable, we can use the
induction assumption to calculate the subtree below each such first
$b$. These subtrees can then be squashed into single leafs using the
antichain composition principle, and therefore the induction
assumption can be used on a smaller alphabet of inner node labels,
which now no longer contains $b$.

We say that two forest types $h,g \in H$ are \emph{$v$-equivalent} if
$vuh=vug$ holds whenever $v$ is not reachable from $vu$.

\begin{lem}\label{lemma:conteqplus}
  For each $h$, the set of forests whose type is $v$-equivalent to $h$
  is forest-definable modulo $X$.
\end{lem}
\begin{proof}
  Fix some context type $u$ such that $v$ is not reachable from $vu$.
  By induction assumption---the third parameter is decreased---the set
  of forests $s$ satisfying $vu\alpha(s)=vuh$ is forest-definable
  modulo $X$. The set in the statement of the lemma is the
  intersection, over $u$, of all these sets.
\end{proof}

\begin{lem}\label{lemma:rgreen-preservers}
  If $v, w, vu \in V$ are in the same context component, then so is $wu$.
\end{lem}
\begin{proof}
  By assumption there must be context types $v',w'$ with $vuw'=w$ and
  $wv'=v$. But then we have $wv'uw'=w$. In particular,
  $w(v'uw')^\omega v' = v$. Using identity~(\ref{eq:da}), we get
  \[
    v = w(v'uw')^\omega v'  = w(v'uw')^\omega uw' (v'uw')^\omega v'
    = wuw' (v'uw')^\omega v'\ ,
  \]
  which shows $v$ can be reached from $wu$.
\end{proof}

Let $\gamma_1,\ldots,\gamma_n$ be all the equivalence classes of
$v$-equivalence. For each such class $\gamma_i$, let $L_i$ be the set
of trees $\set{bt : \alpha(t) \in \gamma_i}$. Thanks to
Lemma~\ref{lemma:conteqplus}, each set $L_i$ is tree-definable. For
any $i=1,\ldots,n$, let $h_i$ be an arbitrarily chosen forest type in
the class $\gamma_i$, and let $a_i$ be a leaf label whose type is
$\alpha(b)h_i$. The label $a_i$ exists by assumption on leaf
saturation.  Note that $a_i$ may have a different type than some of
the trees in $L_i$, since $h_i$ need not be the only forest type in
$\gamma_i$.  However, we will show that no information is lost by
squashing subtree in $L_i$ into a single leaf with label $a_i$, at
least as long as the resulting forest is going to be an argument of
$v$. More formally, we show:
\begin{lem}\label{lemma:v-squash}
  Let $\varphi=b \land \neg \tlf^{-1} b$, i.e.~``a~$b$ without $b$
  ancestors''. For any forest $t$ we have
\[
   v \alpha(t) =  v \alpha(t[ (L_1,\varphi) \to a_1,\ldots,
   (L_n,\varphi) \to a_n])\ .
\]
\end{lem}

Before we show this lemma, we show how it concludes the case
considered in this section. Recall that we want to show that the
following language is forest-definable modulo $X$:
\[
 L=  \set{t : v\alpha(t)=h}\ .
\]
By Lemma~\ref{lemma:v-squash}, this is the same language as
\[
  \set{t : t[ (L_1,\varphi) \to a_1,\ldots,
   (L_n,\varphi) \to a_n] \in L}\ .
\]
Since the  substitution operation removes all
letters $b$ from the forest, we get
\[
L=  \set{t : t[ (L_1,\varphi) \to a_1,\ldots,
   (L_n,\varphi) \to a_n] \in K}\ ,
\]
where $K$ is the set of trees in $L$ that do not use the letter $b$.
To $K$ we can apply the induction assumption on a smaller alphabet,
and then use the antichain composition principle to transfer
definability from $K$ to $L$.

We now resume with the proof of Lemma~\ref{lemma:v-squash}.

\proof
  Note first that the tree on the right hand side of the equation is
  well defined, since the languages $L_i$ are disjoint, and $\varphi$
  is an antichain formula. The proof is by induction on the number of
  $b$ nodes in the forest $t$. The induction base, where there are no
  $b$'s, is immediate since the substitution on the right hand side
  does not change the forest.  Otherwise, let $t$ be of the form
  $pbs$, with the context $p$ not containing any $b$'s on the main path,
  and let $L_i$ be such that $bs \in L_i$.  By induction assumption,
  we have
  \[
    \alpha(pa_i[ (L_1,\varphi) \to a_1,\ldots,
   (L_n,\varphi) \to a_n]) = \alpha(pa_i)\ .
  \] 
  By definition of the substitution we have
  \[
       t[ (L_1,\varphi) \to a_1,\ldots,
   (L_n,\varphi) \to a_n]  =  pa_i[ (L_1,\varphi) \to a_1,\ldots,
   (L_n,\varphi) \to a_n]\ ,
  \]
  it therefore remains to show that $v\alpha(pba_i)=v \alpha(pbs)$.
  
  First, we claim that $v$ is not reachable from $v\alpha(pb)$.
  Indeed, if $v$ is not reachable from $v\alpha(p)$ then we are done.
  Otherwise, $v$ and $v \alpha(p)$ are in the same context component. If this
  context component would also contain $v\alpha(pb)$, then by
  Lemma~\ref{lemma:rgreen-preservers} it would also contain $v
  \alpha(b)$, a contradiction with the assumption on $b$.
  
  Recall now the forest type $h_i$ that represented the equivalence class
  $\gamma_i \ni \alpha(s)$. By assumption on $\alpha(s)$ and $h_i$
  being $v$-equivalent, we get
$$ v\alpha(pbs)=v \alpha(pb)\alpha(s)=v \alpha(pb) h_i = v
    \alpha(p)\alpha(b)h_i = v \alpha(p) \alpha(a_i) = v \alpha(pa_i)\ 
    .\eqno{\qEd}
$$

\subsection{There is more than one forest component in $H \setminus
  X$}\label{sec:one-forest-comp} We now turn to the second case in the
proof of Proposition~\ref{induction:prop}.  Let $G \subseteq H$ be
a forest component not included in~$X$.  We pick $G$ so that no forest
type in $G$ can be reached from a forest type outside $X \cup G$.
Intuitively speaking, forest types from $G$ are close to the leaves.
The essential idea in this section is that we will add $G$ to $X$, by
squashing each subtree of type $g$ to a single leaf with the $g$
written in its label. This is done by applying the antichain
composition.

Let $W \subseteq V$ be the set of context types that preserve $G$,
i.e.~context types $w$ such $g$ is reachable from $wg$ for some $g \in
G$.  The following lemma, proved the same way as
Lemma~\ref{lemma:rgreen-preservers}, shows that ``some'' in the above
definition can be replaced by ``all''.
\begin{lem}\label{lemma:reach-preservers}
  If $g, h, vg$ are in the same forest component, then so is $vh$.
\end{lem}

Let $F \subseteq H$ be the set of those forest types $f$ from which a forest type in
$G$ can be reached. In particular, we have
\[
  G \subseteq F \subseteq H\ .
\]
Note that all forest types in $F \setminus X$ are from $G$ by choice of $G$.
Furthermore, the inclusion $F \subseteq H$ is proper, since $H
\setminus X$ contains more than one forest component by assumption.
The inclusion $G \subseteq F$ may also be proper, however all forest types in
the difference $F \setminus G$ are from $X$.

We say $f \in H$ is a \emph{bad brother} if for all $g \in G$, we have
$f + g \not \in G$, i.e.~$g$ is not reachable from $f+g$. Likewise, we
say $f \in H$ is a \emph{good brother} if for all $g \in G$, we have
$f + g \in G$, i.e.~$g$ is reachable from $f+g$.  Note that by
definition of $F$, all good brothers are in $F$. Clearly $f$ is a bad brother
if and only if the context type $f+\hole$ is outside $W$.  Therefore
by Lemma~\ref{lemma:reach-preservers}, every forest type in $H$ is
either a good brother or a bad brother.  In particular, all forest
types in $G$ are good brothers, since they cannot be bad brothers by
$g + g = g$.  Furthermore, since $W$ is closed under context
composition, good brothers are closed under forest concatenation,
i.e.~form a subsemigroup of $H$.

We fix the sets $G$, $F$ and $W$ for the rest of
Section~\ref{sec:one-forest-comp}.

A \emph{twig} is a tree of depth exactly two, i.e.~a root and some
leaves. A \emph{twig node} is a node whose subtree is a twig.

\begin{lem}\label{lemma:find-h}
  There is a formula $\psi$ such that in any $X$-trimmed tree, $\psi$
  holds in nodes with a subtree of type in $G$.
\end{lem}
\begin{proof}
  Let $t$ be an $X$-trimmed tree, and $x$ a node in this tree.  If the
  node is a leaf, then the type of its subtree can be read from the
  label. Otherwise, the type of the subtree must be either in $G$ or
  outside $F$, by assumption on the tree being $X$-trimmed.  We claim
  that the following condition is necessary and sufficient for the
  subtree of $x$ to have a type outside~$F$, and can furthermore be
  tested by a formula of $\fotwo$.  The condition is that some
  descendant~$y$ of $x$, not necessarily proper, is either
  \begin{enumerate}[(1)]
\item A leaf or twig node with a type outside $F$; or
\item A non-twig inner node with a label $b \in B$ whose type
      $\alpha(b)$ is outside $W$; or
\item An inner node whose brother has a leaf label $a$ whose type
    $\alpha(a)$ is a bad brother.
  \end{enumerate}
  
  We begin by showing that these conditions can be tested by an
  $\fotwo$ formula.  Testing for 1) is simple. Using $\tlef$, we
  search for a candidate $y$ for the node. If $y$ is a leaf, we just
  test its label. Otherwise, we test if $y$ is a twig node (no path of
  length at least two).  Then we read the label of $y$ and the set of
  labels in descendants of $y$, which uniquely determine the type of
  the subtree of $y$, thanks to idempotency and commutativity,
  i.e.~identities~(\ref{eq:apcom}).  Condition 2 is tested in a
  similar way. For condition 3 we use $\tlef$ to go into a leaf $y$
  with a label $a$ whose type is a bad brother.  We then
  test if $y$ has a sibling that is an inner node (all ancestors of
  $y$ have an inner node descendant).
  
  We now show that these conditions are sufficient. The first one is
  clearly sufficient.  For the other two, note that every inner node
  has a subtree with type outside $X$ by assumption on the tree being
  $X$-trimmed.  This type must then be either in $G \supseteq F
  \setminus X$ or outside $F$.  For the second condition, let $bs$ be
  the subtree of a non-twig inner node, with $s$ a forest.  Since $s$
  has depth at least two, its type must be outside $X$, and therefore
  either outside $F$, or in $G \supseteq F \setminus X$. In either
  case, the type of $bs$ is outside $F$. The last condition is shown
  in a similar way.

  It remains to show that the conditions are necessary.  Indeed,
  assume that the subtree of $x$ has a type outside $G$. Let $s$ be a
  minimal subtree below $x$ that has a type outside $F$. If $s$ is a
  leaf or a twig, then item 1 must hold. Otherwise $s$ is of the form
  $b(s_1+\cdots s_n)$, for some label $b \in B$ and trees
  $s_1,\ldots,s_n$, with at least one tree $s_i$ not being a leaf. By
  assumption on the tree being $X$-trimmed, the type of $s_i$ is
  outside $X$.  Since $F \setminus X \subseteq G$, the type of this
  $s_i$ is in $G$.  If all the types of $s_j$, for $j\neq i$, are good
  brothers, then the type of $s_1 + \cdots + s_n$ must belong to $G$
  by closure of good brothers under composition, and therefore case 2
  must hold.  Finally, we consider the case when the type of some tree
  $s_j$ is a bad brother.  Since all forest types from $G$ are good
  brothers, the type of $s_j$ is in $F \setminus G \subseteq X$. Since
  the tree is $X$-trimmed, $s_j$ is a single leaf, and thus 3 holds.
\end{proof}

\begin{lem}\label{lemma:describe-h}
  For each $g \in G$, the set of trees with type $g$ is tree-definable
  modulo $X$.
\end{lem}
The general idea is that $(G,W)$ is a (smaller) forest algebra, and
therefore the induction assumption can be applied to languages
recognized by $(G,W)$.  However, thanks to bad brothers and such,
$(G,W)$ does not recognize the language in the lemma. Before we 
solve this problem, we show how Lemmas~\ref{lemma:find-h} and
\ref{lemma:describe-h} along with the antichain composition principle conclude
the case considered in this section. The idea is that we add all
forest types from $G$ to~$X$.

Let $h,v$ be as in the statement of
Proposition~\ref{induction:prop}. We need to show that the language 
\[
  L = \set{  t :  v(\alpha( t)) = h}
\]
is forest-definable modulo $X$.  By induction assumption, we know that
this language is forest-definable modulo $X \cup G$. In other words,
there is some forest-definable set of forests $K$ that agrees with $L$
over $(X\cup G)$-trimmed forests.  To describe $L$ modulo $X$, we will
use the antichain composition principle.

Let $\psi$ be the formula from Lemma~\ref{lemma:find-h}. Let
\[
  \varphi =   \psi \land \neg \tlf^{-1} \psi\ .
\]
This formula holds in a node whose subtree has a type in $G$, and the
node is closest to the root for this property. Thanks to the last
clause, $\varphi$ is an antichain formula. Let
$G=\set{g_1,\ldots,g_n}$.  By assumption that $\alpha$ is leaf
saturated, for each $g_i$ there is a leaf label $a_i \in A$ with
$\alpha(a_i)=g_i$. For each $g_i$, let $L_i$ be the set of trees with
type $g_i$.  Thanks to Lemma~\ref{lemma:describe-h}, each tree
language $L_i$ is tree-definable modulo $X$.

It is easy to see that squashing a subtree with type $g_i$ into a
single leaf with label $a_i$ does not change the type of the whole
tree. More precisely, a forest $ t$ has the same value as
\[
    t[ (L_1,\varphi) \to a_1,\ldots, (L_n,\varphi) \to a_n] \ .
\]
Furthermore, the above forest is $(X \cup G)$-trimmed, at least as
long as $ t$ was $X$-trimmed.  It follows that over $X$-trimmed
forests, $L$ agrees with
\[
  \set{  t :   t[ (L_1,\varphi) \to a_1,\ldots, (L_n,\varphi)
    \to a_n]  \in K}\ ,
\]
which is forest-definable thanks to the antichain composition
principle. It now remains to show Lemma~\ref{lemma:describe-h}, which
we do in the next section.

\subsubsection{Trees with type in $G$.}
 Fix
some forest type $g \in G$. Our goal is to show that the set of trees with
type $g$ is tree-definable modulo~$X$.

\begin{lem}\label{lemma:loss-of-generality}
  Without loss of generality, we may assume that all forest types in $F$ are
  good brothers and all inner node labels $b$ satisfy $\alpha(b) \in W$.
\end{lem}
\begin{proof}
  Recall that all forest types from $G$ are good brothers. In particular, all
  bad brothers in $F$ are from $X$, and can therefore only appear in
  leaves, as long as we are working over $X$-trimmed forests. Let $A' \subseteq A$
  be the set of leaf labels that are mapped by $\alpha$ to a good
  brother in $F$.  Let $B' \subseteq B$ be the set of inner node
  labels $b$ with $\alpha(b) \in W$.

  Let $\beta$ be the restriction of $\alpha$ to this smaller alphabet:
  \[
    \beta : \freef {A'} {B'} \to (H,V)\ .
  \]
  Note that over $X$-trimmed forests, the only forest types from $F$ in the
  image of $\beta$ are good brothers, and all inner node labels $b$
  satisfy $\alpha(b) \in W$. Assume now, that we have shown
  Lemma~\ref{lemma:describe-h} for the morphism $\beta$, i.e.~the set
  $K$ of trees that have type $g$ under $\beta$ is tree-definable
  modulo~$X$.  We will use the antichain composition principle to extend this
  result to $\alpha$. The idea is that we squash twig nodes into
  leaves, thus eliminating labels outside $A', B'$.

  Let $\varphi$ be a formula that is true in twig nodes (the node is
  not a leaf, but all of its proper descendants are leaves); this is
  clearly an antichain formula.  Let $G=\set{g_1,\ldots,g_n}$.  By
  assumption that $\alpha$ is leaf saturated, for each $g_i$ there is
  a leaf label $a_i \in A$ with $\alpha(a_i)=g_i$. For each $g_i$, let
  $L_i$ be the set of twig trees with value $g_i$ (under $\alpha$).
  Each $L_i$ is tree-definable, since the type of a twig tree is
  determined by its root label and the set of its leaf labels
  by~(\ref{eq:apcom}). It is easy to see that a tree $t$ over $(A,B)$
  has the same type under $\alpha$ as the tree
  \[
    t[ (L_1,\varphi) \to a_1,\ldots, (L_n,\varphi) \to a_n]\ .
  \]
  Furthermore, if the type of $t$ under $\alpha$ is $g$, then the
  latter forest belongs to the domain of $\beta$, since all nodes with
  a label outside $A'$ or $B'$ are covered by $\varphi$.  Therefore, we
  can use the antichain composition principle to conclude that the forests
  with value $g$ under $\alpha$ can be defined in $\fotwo$.
\end{proof}

From now on, we use the assumptions stated in the previous lemma.
Recall that good brothers are closed under concatenation, and
therefore $F$ is a subsemigroup of $H$.  This allows us to define a
semigroup automaton~$\Aa$, whose semigroup is $F$.  The input alphabet
of this automaton is:
\begin{enumerate}[$\bullet$]
\item The inner node labels are $B$
\item The leaf labels are $A'=\set{a \in A : \alpha(a) \in F}$.
\end{enumerate}
For $a \in A'$, we define $\beta_A(a)$ to be $\alpha(a)$.  For $b \in
B$, we would like the associated function $\beta_B(b)$ to be
$\alpha(b)$.  Even though Lemma~\ref{lemma:loss-of-generality}
guarantees that $\alpha(b)$ belongs to $W$, this context type cannot
be used since it need not generate a function $F \to F$. The reason is
that $\alpha(b)h$ may be outside $F$ for types $h$ outside $G$.  To
solve this problem, we artificially redefine the function:
\begin{equation}
  \label{eq:redefine}
  \beta_B(b)(h) = \left\{
  \begin{array}{rcl}
   \alpha(b)h & & \text{ if }\alpha(b)h \in F\\
   g_0  & & \text{ otherwise}.
  \end{array}\right.
\end{equation}
In the above, $g_0$ is an arbitrarily chosen forest type from $G$. 

By the proof of Theorem~\ref{thm:equiv-with-regul}, this automaton
induces a forest algebra morphism
\[
  \beta : \freef {A'} {B'} \to (F,F^F)\ .
\]
This morphism is not the same as $\alpha$, due to the second clause in
(\ref{eq:redefine}). However, it agrees with $\alpha$ over
the forests that are relevant to Lemma~\ref{lemma:describe-h}:
\begin{lem}\label{lemma:new-agrees-with-old}
  For any $g \in G$, and forest $t$, if $\alpha( t)=g$ then $\beta(
  t)=g$.
\end{lem}
\begin{proof}
  If $t$ has a type in $G$ under $\alpha$, then all of its leaf labels
  belong to $A'$ by definition of~$F$. Therefore, $t$ belongs to the
  domain of $\beta$.  The lemma is proved by induction on the size
  of~$t$. If $\alpha( t)=g$, then the ``bad'' second case
  in~(\ref{eq:redefine}) is never used while calculating $\beta( t)$.
\end{proof}
\begin{lem}\label{lemma:new-is-ok}
  The image of $\beta$ satisfies
  identities~(\ref{eq:apcom}),~(\ref{eq:da}) and~(\ref{eq:special}).
\end{lem}
\begin{proof}
  We only focus on identity~(\ref{eq:special}), the others are easy to
  show.  The key idea is that $\alpha$ and $\beta$ only disagree in
  twig nodes, and these are not important for the
  identity~(\ref{eq:special}).

  Let then $p_1 \mainle p_2,q_1 \mainle q_2$ be contexts. 
We need to
  show that
  \[
    \beta((p_1 q_1 )^\omega (p_2 q_2 )^\omega) = 
    \beta(  (p_1 q_1 )^\omega p_1 q_2  (p_2 q_2 )^\omega)\ .
  \]
Thanks to
  the faithfulness of contexts in forest algebra, it suffices to
  show that both sides induce the same transformations on forests, i.e.~
  \[
    \beta((p_1 q_1 )^\omega (p_2 q_2 )^\omega  t) = 
    \beta(  (p_1 q_1 )^\omega p_1 q_2  (p_2 q_2 )^\omega  t)
  \]
  holds for every forest $ t$. 

  Consider first the case when both $p_2,q_2$ have the hole in the
  root, and therefore so do $p_1,q_2$. In this case the equality
  above becomes:
  \[
    \beta(\omega(s_1 + t_1) + \omega(s_2+t_2) + t) =
    \beta(\omega(s_1 + t_1) + s_1 + t_2 + \omega(s_2+t_2) + t)\ .
  \]
  The above equality follows by commutativity of the horizontal monoid
  $F$, and aperiodicity of $H$, i.e.~ $\omega h + h = \omega h$. The
  latter is a consequence of aperiodicity of $V$, itself a consequence
  of (\ref{eq:da}), by iterating
  \[
    \omega h + h = (h+\hole)^\omega h =(h+\hole)^\omega (h+\hole)^\omega h = \omega h + h + h\ .
  \]
  We can therefore now assume that in the context $p_2q_2$, at least
  one inner node is an ancestor of the hole.  Thanks to the assumption
  on leaf saturation, in the contexts $p_1,q_1,p_2,q_2$ every subtree
  that does not contain the hole can be squashed to a single node,
  without affecting the image under $\beta$. We therefore assume that
  in the contexts above, all nodes outside the main path are leaves.
  As remarked above, a consequence of equation~(\ref{eq:da}) is that
  $V$ is aperiodic, i.e.~$v^\omega = v^\omega v$ holds for every
  context type $v$.  Therefore, it is sufficient to show
  \begin{equation}\label{eq:redefined-ok}
    \begin{array}{rcl}
          \beta((p_1 q_1 )^\omega (p_2 q_2 )^\omega(p_2 q_2 )  t)=
    \beta(  (p_1 q_1 )^\omega p_1 q_2  (p_2 q_2 )^\omega (p_2 q_2 )
     t)\ .
    \end{array}
  \end{equation}
  The only part where $\alpha$ and $\beta$ disagree are twig nodes.
  Thanks to our assumption on the form of $p_1,p_2,q_1,q_2$, the only
  place where the forests in~(\ref{eq:redefined-ok}) contain twig
  nodes is $p_2q_2 t$.  Therefore, we have
\[
  \beta((p_1 q_1 )^\omega (p_2 q_2 )^\omega(p_2 q_2 )  t) =
  \alpha((p_1 q_1 )^\omega (p_2 q_2 )^\omega) \beta((p_2 q_2 ) 
  t)\ .
\]
In the same way we can decompose the right side
of~(\ref{eq:redefined-ok}). Applying the assumption that the image of
$\alpha$ satisfies~(\ref{eq:special}), we get the desired result.
\end{proof}

\begin{proof}[Proof of Lemma~\ref{lemma:describe-h}]
  By Lemma~\ref{lemma:new-agrees-with-old}, a tree has type $g$ under
  $\alpha$ if and only if a) its type under $\alpha$ belongs to $G$;
  and b) it has type $g$ under $\beta$. Condition a) can be tested by
  thanks to Lemma~\ref{lemma:find-h}.  Since $F$ is a proper subset of
  $H$, we can use the induction assumption to test condition b).
\end{proof}

\subsection{The induction base}\label{sec:otherwise}

\newcommand{\efeq}{\equiv}
In this section, we assume that the techniques from the previous two
sections cannot be applied. That is:
\begin{enumerate}[$\bullet$]
\item All forest types from $H \setminus X$ are in a single forest
  component.
\item For all inner node labels $b \in B$, $v$ is reachable from
  $v\alpha(b)$.
\end{enumerate}
Note that the second assumption does not necessarily mean that any
context type reachable from $v$ is in the same context component.
Indeed, it is possible that for some forest type $g$, the context type
$v$ is no longer reachable from $v(\hole + g)$.

We will show
\begin{equation}
  \label{eq:one-component}
  vf=vg \qquad \mbox{for all } f,g \in H
  \setminus X\ .
\end{equation}
Before we do this, we show how Proposition~\ref{induction:prop}
follows. For every every forest type $h \in H$, we need to show that
the forest language
\[
L =  \set{  t :  v(\alpha( t)) = h}
\]
is forest definable modulo $X$.  By
assumption~(\ref{eq:one-component}), there is some forest type $h_0
\in H$ such that $vf=h_0$ holds for all $f \in H \setminus X$.
\begin{enumerate}[$\bullet$]
\item If an $X$-trimmed forest $ t$ contains an inner node label---which can easily be tested by the logic---then $\alpha( t)$
  must be in the single forest component $H \setminus X$. In
  particular, $v\alpha( t)=h_0$.  So in this case, $\varphi$ is
  either ``true'' or ``false'' depending on whether $h_0=h$ or not.
\item Otherwise, the forest $ t$ is the concatenation of some leaves
      $a_1+\cdots +a_n$. In this case, the type of $v\alpha(t)$ can be
      calculated based on the set of leaf labels in $ t$.
\end{enumerate}

The rest of this section is devoted to showing
(\ref{eq:one-component}).  The following lemma is the key step in our
proof~(\ref{eq:one-component}). It says that not only any two forest
types $h,g \in H \setminus X$ can be reached from each other---which is
the assumption on there being one forest component---but they can also
be reached from each other by only using contexts without any
branching. Furthermore, the context type that goes from $g$ to $h$ can
be chosen independently of $g$. However, all these statements are
relative to context types from the context component of $v$.
\begin{lem}\label{lemma:letter-reach}
  Let $h \in H \setminus X$.  There are inner node labels
  $b_1,\ldots,b_n \in B$ such that $wh=w\alpha(b_1 \cdots b_n)g$ holds
  for each forest type  $g \in H \setminus X$ and context type $w$ in the context
  component of $v$.
\end{lem}
\begin{proof}
  Let $h$ be a forest type outside $ X$.  We first show that there is
  a context type $u_g$ such that $h=u_hf$ holds for every forest type
  $f \in H$. By assumption on there being only one forest component
  outside $X$, the forest type $h$ can be reached from every forest
  type. In particular, there is some context type $u$ such that
  $h=u(h_1+\cdots+h_n)$, where $h_1,\ldots,h_n$ are all the forest
  types in $H$. Let
  \[
    u_h = u(h_1+\cdots+h_n + \hole)  \ .
  \]
  Thanks to idempotency and commutativity of $H$,
  i.e.~identity~(\ref{eq:apcom}), 
  \[
    h_1+\cdots +h_n =     h_1+\cdots +h_n + f
  \]
  holds for any forest type $f$, and therefore also $h=u_hf$.

   We can decompose the context  $u_h$ as 
  \[
    u_h=(f_1+\alpha(b_1)) \cdots (f_n + \alpha(b_n)) 
  \]
  for some $n $ and $f_1,\ldots,f_n \in H$ and $b_1,\ldots,b_n
  \in B$. (In general, some of the $f_i$ may be empty; but the proof
  follows the same lines.)  Let us denote $\alpha(b_i)$ by $v_i$. We
  will show that
  \[
    wh=wv_1\cdots v_ng 
  \]
  holds for any forest type $g$ and any context type $w$ in the
  context component of $v$, thus proving the lemma.

  Let then $g,w$ be as above. As for $h$, we can define a context type
  $u_g$ such that $g=u_gf$ holds for any forest type $f$. This context
  can also be decomposed as
  \[
        u_g=(f_{n+1}+\alpha(b_{n+1})) \cdots (f_m + \alpha(b_m))
  \]
  for some $m \ge n+1 $ and $f_{n+1},\ldots,f_m \in H$ and
  $b_{n+1},\ldots,b_m \in B$.  As previously, we denote $\alpha(b_i)$
  by $v_i$.  By definition, we have
   \begin{eqnarray}\label{eq:mainle}
     v_1 \cdots v_n \mainle u_h \qquad v_{n+1} \cdots v_m \mainle u_g
   \end{eqnarray}
   
   Let now $w \in V$ be in the same context component as $v$.  By
   assumption on $w$ and Lemma~\ref{lemma:rgreen-preservers}, also the
    context type $w v_1 \cdots v_m$ is in the same context component as $v$.
  In particular, there is some $ \bar w \in V$ such that 
  \[
        w v_1 \cdots v_m  \bar w = w\ .
  \]
  By iterating the above $\omega$ times, and appending $h$, we get
\[wh =   w(v_1 \cdots v_m  \bar  w)^\omega h\ .\]
Since $u_hf=h$ holds for all forest types $f$, the above can be
rewritten as
\[w(v_1 \cdots v_m  \bar w)^\omega (u_hu_g  \bar w)^\omega h\ .\]
Using the property from identity~(\ref{eq:special}), we get
\[\eqalign{
   w(v_1\cdots v_m\bar w)^\omega (u_h u_g\bar w)^\omega h
&= w(v_1\cdots v_m\bar w)^\omega v_1 \cdots v_n u_g\bar w (u_hu_g 
   v)^\omega h\cr
&= w(v_1 \cdots v_m  \bar w)^\omega v_1 \cdots v_n g=
\bar wv_1\cdots v_n g\ ,
  }
\]
which concludes the proof of the lemma.
\end{proof}

We now use the above Lemma to conclude the proof of
(\ref{eq:one-component}). Indeed, let $f,g$ be forest types outside
$X$. By the above lemma, there are inner node labels $b_1,\ldots,b_m
\in B$ such that
\[
  f=w\alpha(b_1 \cdots b_n)h \qquad   g=w\alpha(b_{n+1} \cdots b_m)h
\]
holds for all $w$ in the context component of $v$ and all forest types
$h$ outside $X$.  Let $v_i=\alpha(b_i)$. By assumption on the
equivalence class of $v$ and by Lemma~\ref{lemma:rgreen-preservers},
there must be some $ v \in V$ such that
\[v v_1 \cdots v_m  \bar v = v\ .\]
But then we have
\[  vf = v(v_1 \cdots v_m  \bar v)^\omega f = 
    v(v_1 \cdots v_m  \bar v)^\omega v_{n+1} \cdots v_m  \bar v (v_1
  \cdots v_m  \bar v)^\omega  f = 
    v(v_1 \cdots v_m  \bar v)^\omega g = vg\ .
\]
The second equality follows from (\ref{eq:da}).

\section{Empty forests}
\label{sec:empty-forests}
The forest algebra setting used in this paper does not allow empty
forests.  There is also a two-sorted alphabet $(A,B)$, where letters
from $A$ are only allowed in leaves, and letters from $B$ are only
allowed in inner nodes.  A different, and arguably more elegant,
setting is considered in~\cite{forestalgebra}, where empty forests are
allowed, and only one alphabet is used. 

Why do we not use the forest algebra with empty forests here? The
reason is that the completeness proof in
Proposition~\ref{induction:prop} uses an induction on the size of the
leaf alphabet, so it helps that the leaf alphabet is part of the
definition of the forest algebra. The assumption on nonempty forests
follows, since if we want a separate alphabet for leaves, there are
algebraic reasons to consider forest algebras without the empty
forest.  A natural question emerges: does our characterization also
work for forest algebra with empty forests?  In this section, we give
an informal argument that the answer to this question is yes.

We will not give a detailed discussion of forest algebra with empty
forests here. We define only define the syntactic object. The
interested reader is referred to~\cite{forestalgebra}.  Let $A$ be an
alphabet. We define $A^\Delta_H$ (respectively, $A^\Delta_V$) to be
the set of (possibly) empty forests (respectively, contexts) labeled
by $A$, without any restriction on labels in leaves or inner nodes. We
write $\onefreef A$ for the pair $(A^\Delta_H,A^\Delta_V)$. The only
difference between $\onefreef A$ and $\freef A A$ is that the second
does not allow the empty forest on its first coordinate.  It is not
hard to see that $\onefreef A$ is a forest algebra, as defined in
Section~\ref{sec:forest-algebra}.  Given a set $L$ of forests,
possibly including the empty forest, the \emph{syntactic forest
  algebra with empty forests of $L$} is defined to be the quotient of
$\onefreef A$ under the two-sorted equivalence relation defined below.
\[\eqalign{
  t &\simeq t'\cr
  p &\simeq p'\cr
  }
  \qquad
  \eqalign{
  \miff  &\qquad \forall p \in A^\Delta_V  \ \ pt \in
  L\  \iff\  pt' \in L\cr
  \miff  &\qquad \forall q \in A^\Delta_V \forall s
  \in A^\Delta_H  \ \ qps \in L\  \iff\  qp's \in L\cr
  }
\]
This equivalence relation is a refinement of the Myhill-Nerode
equivalence introduced in Section~\ref{sec:forest-algebra} (for the
case when $A=B$).  It may possibly distinguish more contexts because
the variable $s$ can also quantify over the empty forest.

\begin{thm}
  Let $L$ be a forest language. Let $(H,V)$ be its syntactic forest
  algebra, and let $(H',V')$ be its syntactic forest algebra with
  empty forests. If $(H,V)$ satisfies the identities from
  Theorem~\ref{thm:main}, then so does $(H',V')$, and vice versa.
\end{thm}
\begin{proof}
  We begin with the right to left implication. Since $\freef A A$ is a
  subalgebra of $\onefreef A$, and since the equivalence relation
  defining $(H',V')$ is a refinement of the equivalence relation
  defining $(H,V)$, it follows that $(H,V)$ is a subalgebra of
  $(H',V')$. In particular, any identities that hold in the latter
  must also hold in the former.

  For the left to right implication, assume that $(H,V)$ satisfies the
  identities from Theorem~\ref{thm:main}. By the theorem, the
  recognized language $L$ is forest-definable in $\fotwo$. To
  conclude, we will show that if a language $L$ is forest-definable in
  $\fotwo$, then its syntactic forest algebra with empty forests
  $(H',V')$ satisfies the identities from Theorem~\ref{thm:main}. This
  follows by the correctness argument presented in
  Section~\ref{sec:correctness}. The reason why we can use that
  argument is that it relied on Proposition~\ref{prop:ef-transfer} to
  transfer Duplicator strategies, and this proposition also works for
  the more general one-sorted morphisms that are appropriate for
  forest algebras with empty forests.
\end{proof}

\section{One quantifier alternation}
In~\cite{therienwilkefo2}, it was shown that over words, the temporal
logic $\wordfotwo$ has the same expressive power as $\Sigma_2 \cap
\Pi_2$, where
\begin{enumerate}[$\bullet$]
\item $\Sigma_2$ are word properties definable by a first-order
  formula with quantifier prefix $\exists^* \forall^*$; the signature
  contains label tests and the left-to-right order on word positions.
\item $\Pi_2$ are complements of $\Sigma_2$.
\end{enumerate}
For instance, consider the word language $b^*aA^*$ over the alphabet
$A=\set{a,b,c}$. This language can be defined in $\wordfotwo$ by the
formula
\[
\tlf (a \land \neg \tlf^{-1} \neg b)\ .
\]
This language can also be defined both in $\Sigma_2$ and $\Pi_2$, as witnessed by the formulas:
\[
  \begin{array}{lll}
      \exists x \forall y \  &\ a(x) \ \land \  (y < x\ \Rightarrow \ b(y)) & \qquad \in \Sigma_2\\
  \forall x \exists y \  &\ c(x)\ \Rightarrow \ (y < x \ \land \ a(y))& \qquad \in \Pi_2\ .
  \end{array}
\]
Both classes $\Sigma_2$ and $\Pi_2$ can be extended to trees using the
descendant order on tree nodes.  We show here that the result
from~\cite{therienwilkefo2} fails for trees:
\begin{prop}
  Over trees, the classes $\fotwo$ and $\Sigma_2 \cap \Pi_2$ have
  incomparable expressive power. Likewise for forests.
\end{prop}

A mentioned in the introduction, the class $\Sigma_2 \cap \Pi_2$ was
given an effective characterization in~\cite{deltatwo}. We prove the
above proposition for forests, the case for trees is done the same
way. The inequality
\[
  \fotwo \supsetneq \Sigma_2 \cap \Pi_2
\]
is witnessed by the language ``three nodes with label $a$'', which
cannot be defined in $\fotwo$ by virtue of~(\ref{eq:apcom}). To show
the remaining inequality
\[
  \fotwo \subsetneq \Sigma_2 \cap \Pi_2\ ,
\]
we will demonstrate in the following lemma that the forest property
``no root node is a leaf'' cannot be defined in $\Sigma_2$, although
it is forest-definable in $\fotwo$.
\begin{lem}\label{lemma:delta2-pumping}
  Let $a$ be a leaf label, $b$ an inner node label, and $\varphi$ be a
  formula of the form
\[
  \exists x_1 \ldots x_i \forall y_1 \ldots y_j \psi(x_1\ldots x_i,
  y_1\ldots y_j) \quad \in \Sigma_2\ ,
\]
with $\psi$ quantifier-free.  Let $n>i+j$. If $n(ba)$ satisfies $\varphi$, then so does $n(ba)+a$.
\end{lem}
\begin{proof}
  Assume then that $n(ba)$ satisfies $\varphi$. We need to show that
  $n(ba) + a$ does too. For $x_1,\ldots,x_i$, we pick the same nodes
  in $n(ba)+a$ as the nodes in $n(ba)$ that witnessed $\varphi$.  We
  need to show that for any assignment of the nodes $y_1,\ldots,y_j$
  in $n(ba)+a$ that makes $\psi$ false, we also can find an assignment
  in $n(ba)$ that makes $\psi$ false.  The key point is that any
  assignment of $x_1,\ldots,x_i,y_1,\ldots,y_j$ in $n(ba)+a$ must
  leave at least one copy of $ba$ without any variables; this copy can
  be used in $n(ba)$ to simulate $a$.
\end{proof}

\section{Closing remarks}\label{sec:conclusions}
The contribution of this paper is a characterization of languages
definable in $\fotwo$. This characterization is expressed in terms of
identities that must be satisfied in the syntactic algebra. A
corollary of this characterization is an algorithm for deciding if a
given regular language can be expressed in $\fotwo$. The algorithm
runs in polynomial time if the input is given as a forest algebra.

As mentioned in the introduction, there are many open problems waiting
to be solved in this field. Of those closely related to $\fotwo$, the
following look interesting:
\begin{enumerate}[$\bullet$]
\item What are the identities for two-variable first-order logic with
  the descendant relation? The question boils down to: what identity
  should replace idempotency $h+h=h$? Here is one candidate: $ v(h+h)
  + vh = vh + vh$.
\item What are the identities for an extension of $\fotwo$, where we
  allow operators of the form $\tlef^{k} \varphi$, with the meaning:
  ``the current node has $k$ incomparable descendants where $\varphi$
  holds''. This seems to be a reasonable extension of $\fotwo$ that is
  capable of counting in a proper way (recall that two-variable logic
  could express the property ``there are two $a$'s'', but not the
  property ``there are three $a$'s'').
\end{enumerate}
It is conceivable that a modification of the techniques developed in
this paper can be sufficient to solve the above two logics. For other
logics mentioned in this paper, such as full first-order logic, or
even variants of $\fotwo$ with horizontal order, new techniques need
to be developed.

\end{document}